\documentclass[sigplan,9pt, nonacm]{acmart}

\usepackage[utf8]{inputenc}

\AtBeginDocument{%
  \providecommand\BibTeX{{%
    \normalfont B\kern-0.5em{\scshape i\kern-0.25em b}\kern-0.8em\TeX}}}

\setcopyright{acmcopyright}
\copyrightyear{2020}
\acmYear{2020}
\acmDOI{TBA}

\acmConference[Preprint]{}{}{}

\usepackage{bm}
\usepackage{epstopdf} 
\usepackage{graphicx}

\newcommand{\R}{\mathbb{R}}

\def\cLo{\mathcal{L}^0}
\def\cLi{\mathcal{L}^1}
\def\cN{\mathcal{N}}

\def\sp{\mathcal{P}_{i,j}^\text{min}}

\begin{document}

\title{Transmission Expansion Planning Using Cycle Flows}

\author{Fabian Neumann}
\email{fabian.neumann@kit.edu}
\orcid{0000-0001-8551-1480}
\author{Tom Brown}
\email{tom.brown@kit.edu}
\orcid{0000-0001-5898-1911}
\affiliation{%
  \institution{Institute for Automation and Applied Informatics,\\Karlsruhe Institute of Technology (KIT)}
  \streetaddress{Hermann-von-Helmholtz-Platz 1}
  \city{Eggenstein-Leopoldshafen}
  \country{Germany}
  \postcode{76344}
}


\begin{abstract}
  The common linear optimal power flow (LOPF) formulation that underlies most
  transmission expansion planning (TEP) formulations
  uses bus voltage angles as auxiliary optimization variables
  to describe Kirchhoff's voltage law.
  As well as introducing a large number of auxiliary variables, the angle-based formulation
  has the disadvantage that it is not well-suited to considering the connection of multiple
  disconnected networks,
  It is, however, possible to circumvent these auxiliary variables and reduce
  the required number of constraints by expressing Kirchhoff's voltage law
  directly in terms of the power flows, based on a cycle decomposition of the network graph.
  In computationally challenging benchmarks such as generation capacity expansion with multi-period LOPF,
  this equivalent reformulation was shown in previous work to
  reduce solving times for LOPF problems by an order of magnitude.
  Allowing line capacity to be co-optimized in a discrete TEP problem
  makes it a non-convex mixed-integer problem.
  This paper develops a novel cycle-based reformulation for the TEP
  problem with LOPF and compares it to the standard angle-based formulation.
  The combinatorics of the connection of multiple disconnected networks is formalized for both formulations, a topic
  which has not received attention in the literature.
  The cycle-based formulation is shown to conveniently accommodate synchronization options.
  Since both formulations use the big-$M$ disjunctive relaxation, useful derivations
  for suitable big-$M$ values are provided.
  The competing formulations are benchmarked on a realistic generation and transmission expansion model
  of the European transmission system at varying spatial and temporal resolutions.
  The cycle-based formulation solves up to 31 times faster for particular cases, while averaging at a speed-up of factor 4.
\end{abstract}



\maketitle

\section{Introduction}

Rising shares of renewable energy have put transmission grids under
strain in recent years. The connection of wind turbines to the grid
far from demand has led to frequent situations of high network loading
in countries such as Denmark, Germany and the United Kingdom,
resulting in high levels of wind curtailment. Grid planners must
consider where to reinforce the network in a way that reduces overall
system costs, while also taking account of landscape and environmental
impacts \cite{lumbreras_challenges}.

Transmission Expansion Planning (TEP) is the process of optimizing the
addition of new transmission lines to an existing network. Large
shares of weather-dependent renewables mean that investments need to
be optimized over many representative weather and load conditions,
which drives up the computational burden of TEP in the presence of
renewables.

A common approach to TEP in the literature is to linearize the power
flow equations, which allows TEP problems to be written as mixed
integer linear problems (MILPs) and solved in reasonable time using
decompositions methods and specialized commercial solvers
\cite{binato_benders,gtep,romero,lumbreras_faster}. Such
approaches introduce auxiliary variables for the voltage angles to
formulate the linearized power flows. The use of voltage angles has
two major drawbacks: it introduces many new variables and constraints,
which can lead to performance problems, and it is difficult to
consider the connection of multiple disconnected networks. The latter
difficulty is due to the fact that the voltage angles are only defined
up to a constant in each connected network, and this redundancy must
be managed with care when changing the connectivity. The connection of
previously-disconnected networks is relevant for the connection of
island systems and regions with multiple synchronous zones, like Europe,
North America, China and Japan.

An alternative formulation of the linearized power flow equations has
recently been used for linear optimal power flow (LOPF) problems
without TEP that uses constraints imposed directly on the power flows
themselves, without the use of auxiliary variables, using a cycle
decomposition of the flow pattern \cite{cycleflows}. This cycle-based
formulation was shown to reduce computation times by an order of
magnitude compared to the angle-based formulation in LOPF problems
with generation capacity expansion.

The cycle-based formulation has previously been applied to
the optimal transmission switching (OTS) problem which is related to the TEP problem \cite{transmission_switching}.
OTS is an operational problem where the network topology can be changed by switching lines on and off.
In many regards OTS could be viewed as reverse TEP.
However, using a cycle-based formulation in TEP has a distinct advantage over using it in OTS:
while OTS needs to consider all simple cycles \cite{transmission_switching},
TEP can be formulated by supplementing the initial cycle basis with new candidate cycles
because existing lines are not removed.

In this paper the cycle-based formulation is extended to TEP problems.
It is shown how to choose the big-$M$ parameters necessary for the
disjunctive relaxation, which is also present in the angle-based formulation.
This is important because previous studies have reported a large impact of
big-$M$ coefficients on computation times \cite{lumbreras_realsized}.
For both formulations, it is shown how to formulate problems where
multiple disconnected networks (also called synchronous zones) may be
connected, which involves managing the choice of big-$M$ parameters
and, in the case of the angle-based formulation, the relaxation of the
slack voltage angle constraints. The connection of networks is found to
be both easier to formulate and faster to solve for the cycle-based
formulation.

Realistic benchmark cases with varying spatial and temporal resolution
are provided using the open model dataset PyPSA-Eur \cite{pypsaeur,spatial_scale}.
The model covers the European transmission system and includes regionally resolved
time series for renewable generator availability and is used to formulate
a coordinated expansion planning problem of generation and transmission infrastructure.
All formulations have been implemented for the power system analysis toolbox PyPSA \cite{pypsa}.

The remainder of the paper is structured as follows.
Section \ref{sec:lopf} guides through the foundations of angle-based and cycle-based linear power flow constraints, which are subsequently adapted to the TEP problem in Section \ref{sec:tep}.
The competing TEP formulations are benchmarked and assessed in Sections \ref{sec:setup} and \ref{sec:results}, before conclusions are drawn in Section \ref{sec:conclusion}.

\section{Linear Optimal Power Flow Formulations}
\label{sec:lopf}

Linear optimal power flow (LOPF) problems typically optimize the
dispatch of generation assets in a network with the objective to
minimize costs at the same time as enforcing the physical flow of
power using the linear approximation of the power flow equations.
More general problems consider multiple time periods, so that storage
assets can be optimized as well as investments in assets taking into
account representative load and weather situations.

This contribution considers
long-term investment planning problems that seek to
find cost-effective solutions to reduce greenhouse gas emissions in the power system, of which LOPF is
a one principal building block.
The objective is to minimize the total annual system costs of the network, comprising annualized\footnote{The annuity factor $\left(1-(1+r)^{-n}\right)r^{-1}
	$ converts the overnight investment of an asset to annual payments considering its
	lifetime $n$ and cost of capital $r$.} capital costs $c_*$ for
capacity expansion of generators $G_{i,s}$ and storage $H_{i,s}$  at nodes $i$, and transmission infrastructure $F_{\ell,s}$
at edges $\ell$ of technology $s$, as well as the variable operating costs $o_*$ for the generator dispatch $g_{i,s,t}$
\begin{align}
	\min_{G,H,F,g} \; f(G,H,F,g) \;=\; \min_{G,H,F,g} \quad \left[\sum_{i,s} c_{i,s}G_{i,s} +  \right. \nonumber \\
	\left. \sum_{i,s} c_{i,s}H_{i,s} + \sum_{\ell,s}c_{\ell,s}F_{\ell,s}+\sum_{i,s,t}w_t o_{i,s} g_{i,s,t} \right],
	\label{eq:objective}
\end{align}
where representative snapshots $t$ are weighted by $w_t$ such that their total duration
accumulates to one year \cite{pypsa,pypsaeur}.

The cost-minimizing objective is subject to a set of linear constraints that define limits on
(i) the capacities of generation, storage and transmission infrastructure from geographical and technical potentials,
(ii) the availability of variable renewable energy sources for each location and point in time derived from re-analysis weather data,
(iii) the budget of greenhouse-gas emissions,
(iv) storage consistency equations, and
(v) a multi-period LOPF formulation which, among others, constrains
the absolute active power flows $f_\ell^0$ in all existing lines $\ell\in \cLo$ to remain within their nominal capacities $F_{\ell}^0$
\begin{equation}
	|f_{\ell}^0| \leq F_{\ell}^0.
	\label{eq:existinglinelimit}
\end{equation}
The label $0$ indicates lines in the existing network.

Kirchhoff's Current Law (KCL) and Kirchhoff's Voltage Law (KVL) govern the flow $f_\ell^0$.
A variety of mathematically equivalent LOPF formulations exists,
many of which were compared and benchmarked in previous work \cite{cycleflows}.
In continuous linear capacity expansion problems without discrete transmission expansion planning
the choice of the LOPF formulation was shown to have a great impact on computation times.

In preparation for their extension to discrete transmission expansion planning in subsequent Section \ref{sec:tep},
this section reviews two LOPF formulations used in this setting.
These are (i) the common angle-based formulation using voltage angles as auxiliary variables (cf.\ Section \ref{sec:ab-kvl}) and
(ii) a more efficient cycle-based formulation deduced from graph-theoretical considerations (cf.\ Section \ref{sec:cb-kvl}).
Both formulations share the constraints for representing KCL (cf.\ Section \ref{sec:kcl}), but differ in their formulation of KVL.
While the former has previously been widely used in TEP studies \cite{binato_benders,gtep,romero}, the application of the latter is a novel contribution of this paper.

\begin{table}
	\def\arraystretch{1.1}
	\begin{tabular}{@{}ll@{}}
		\toprule
		\textbf{Symbol}                 & \textbf{Description}                                           \\
		\midrule
		$\cN$                           & set of buses                                                   \\
		$\cN_0$                         & set of slack buses (reference buses)                           \\
		$\cLo$                          & set of existing lines                                          \\
		$\cLi$                          & set of candidate lines                                         \\
		$\cLi_{\text{intra}}$           & set of candidate lines within synchronous zone                 \\
		$\cLi_{\text{inter}}$           & set of candidate lines across synchronous zones                \\
		$\cLi_i$                        & set of candidate lines relaxing slack $\theta_i\;|\;i\in\cN_0$ \\
		$\mathcal{S}$                   & set of synchronous zones                                       \\
		\midrule
		$K_{i\ell}$                     & incidence matrix for lines $\ell$ at buses $i$                 \\
		$B_{\ell\ell}$                  & diagonal susceptance matrix of lines $\ell$                    \\
		$L_{ij}$                        & weighted Laplacian ($L=KBK^\top$)                              \\
		$C_{\ell c}^0$                  & cycles basis matrix of existing network                        \\
		$C_{\ell c}^1$                  & candidate cycle matrix                                         \\
		\midrule
		$p_i$                           & power injection at node $i$                                    \\
		$f_\ell^{0/1}$                  & power flow in existing/candidate line $\ell$                   \\
		$\theta_\ell=\theta_i-\theta_j$ & voltage angle difference between buses $i$ and $j$             \\
		$i_\ell$                        & binary line investment variable ($i_\ell \in \mathbb{B}$)      \\
		\midrule
		$F_\ell^{0/1}$                  & nominal capacity of existing/candidate line $\ell$             \\
		$x_\ell^{0/1}$                  & series reactance of existing/candidate line $\ell$             \\
		$\sp$                           & shortest path between buses $i$ and $j$                        \\
		$M_\ell^\text{KVL}$             & Big-$M$ parameter for angle-based power flow                   \\
		$M_c^\text{KVL}$                & Big-$M$ parameter for cycle-based power flow                   \\
		$M_\ell^\text{slack}$           & Big-$M$ parameter for slack constraints                        \\
		\bottomrule
	\end{tabular}
	\caption{Nomenclature}
\end{table}

\subsection{Kirchhoff's Current Law (KCL)}
\label{sec:kcl}

Kirchhoff's Current Law (KCL) requires the power injected at each bus to
equal the power withdrawn by attached lines; i.e.
\begin{equation}
	p_i = \sum_\ell K_{i\ell} f_\ell^0 \qquad \forall i \in \cN
	\label{eq:kcl}
\end{equation}
where $p_i$ is the active power injected or consumed at node $i\in \cN$,
$f_\ell^0$ is the active power flow on line $\ell$, and
$K\in\R^{|\cN|\times |\cLo|}$ is the incidence matrix of the network graph
which has non-zero values $+1$ if line $\ell$ starts at bus $i$ and $-1$ if line $\ell$ ends at bus $i$.
The orientation of lines is arbitrary but fixed \cite{dual_line_outages}.

KCL provides $|\cN|$ linear equations for the $|\cLo|$ unknown flows $f_\ell^0$, of which one is linearly dependent \cite{cycleflows}.
If the network is a tree with $|\cLo|=|\cN|-1$, equation \eqref{eq:kcl} is already sufficient to uniquely determine the flows $f_\ell^0$.
However, in meshed networks $|\cLo|-|\cN|+1$ additional independent equations are required.
These are provided by Kirchhoff's Voltage Law (KVL).

\subsection{Angle-based Kirchhoff's Voltage Law (KVL)}
\label{sec:ab-kvl}

In textbooks and software toolboxes, Kirchhoff's Voltage Law (KVL) for the linearized power flow is commonly formulated in terms of the voltage phase angles $\{\theta_i\}_{i\in\cN}$ \cite{grainer,convex_optimization}.
This angle-based formulation originates directly from applying the assumptions for linearized power flow to the nonlinear power flow equations in voltage-polar coordinates of lines $\ell\in\cLo$
\begin{align}
	f_\ell^0 = p_\ell & = g_\ell |V_i|^2 - |V_i||V_j| (g_\ell\cos(\theta_i-\theta_j)-b_\ell\sin(\theta_i-\theta_j))  \\
	q_\ell            & = b_\ell |V_i|^2 - |V_i||V_j| (g_\ell\sin(\theta_i-\theta_j)-b_\ell\cos(\theta_i-\theta_j)).
\end{align}
Assuming (i) all voltage magnitudes $|V_i|$ are close to one per unit,
(ii) conductances $g_\ell$ are negligible relative to the susceptances $b_\ell$,
(iii) voltage angle differences are small enough such that $\sin(\theta_i-\theta_j)\approx\theta_i-\theta_j$, and
(iv) reactive power flows $q_\ell$ are negligible compared to real power flows $p_\ell$ leads to
\begin{equation}
	f_\ell^0 = \frac{\theta_\ell}{x_\ell^0}=\frac{1}{x_\ell^0} \sum_i K_{i\ell} \theta_i \qquad \forall \ell \in \cLo
	\label{eq:angle-kvl}
\end{equation}
where $x_\ell^0=b_\ell^{-1}$ is the line reactance and $\theta_\ell=\theta_i - \theta_j$ is the voltage angle difference
between nodes $i$ and $j$ which line $\ell$ connects \cite{convex_optimization}.

Additionally, a reference voltage angle is commonly set at one bus for each synchronous zone
\begin{equation}
	\theta_i = 0 \qquad \forall i \in \cN_0
	\label{eq:slack}
\end{equation}
where $\cN_0$ denotes the set of slack buses.
This circumvents the rotational degeneracy\footnote{The term degeneracy is used to describe the condition where different values for optimization variables yield same optimal objective value. Degeneracy is known to have a detrimental impact on the convergence of both simplex and interior-point methods.}
that originates from the invariance of the network flows to adding a constant to all voltage angles $\theta_i \rightarrow \theta_i + c$ \cite{molzahn_survey}.
Together with the KCL constraints, the angle-based formulation provides $|\cLo|+|\cN|$ independent equality constraints to determine the $|\cLo|$ flows and $|\cN|$ angles.

\subsection{Cycle-based Kirchhoff's Voltage Law (KVL)}
\label{sec:cb-kvl}

KVL states that the sum of voltage angle differences across lines around all cycles in the network must sum to zero.
This allows a reformulation of the linearized power flow equations which circumvents the auxiliary voltage angle variables.
The consistency of voltage angle summations within a connected network can alternatively be achieved by using a cycle basis of the network graph $\mathcal{G}=(\cN, \cLo)$.
A cycle basis is a subset of all simple cycles of $\mathcal{G}$ such that any other cycle can be described
by a linear combination of cycles in the cycle basis \cite{biggs_graph_theory,kavitha_cycle_bases}.
It can be constructed from a minimum spanning tree $\mathcal{T}$ of the network graph in $\tilde O(|\cN|\cdot|\cLo|^2)$ \cite{minimum_cycle_basis}.
The tree $\mathcal{T}$ has $|\cN|-1$ edges \cite{bollobas_graph_theory}.
Together with the path in $\mathcal{T}$ connecting their nodes, each of the $|\cLo|-|\cN|+1$ remaining edges of $\mathcal{G}$ creates a cycle of the cycle basis.
These cycles are linearly independent because each cycle contains an edge that is not contained in the other cycles and
consequently constitute a basis of the cycle space of $\mathcal{G}$ \cite{dual_ptdf}.
These are sufficient to express KVL and uniquely determine the flows $f_\ell^0$ \cite{manik_cycleflows}.
The independent cycles $c\in\{1,\dots,|\cLo|-|\cN|+1\}$ are expressed as a directed linear combination of the lines $\ell$ in the cycle incidence matrix
\begin{equation}
	C_{\ell c}^0 = \left\{
	\begin{array}{r l}
		1   & \; \mbox{if edge $\ell$ is element of cycle $c$},          \\
		- 1 & \; \mbox{if reversed edge $\ell$ is element of cycle $c$}, \\
		0   & \; \mbox{otherwise}.
	\end{array} \right.
	\label{eqn:cycle-edge-matrix}
\end{equation}
Then KVL can be written as
\begin{equation}
	\sum_{\ell} C_{\ell c}^0 \theta_{\ell} = 0 \hspace{1cm} \forall c=1,\ldots,|\cLo|-|\cN|+1.  \label{eq:cycle}
\end{equation}
where $\theta_\ell = \theta_i - \theta_j$ is the angle difference between the two nodes $i$ and $j$ which line $\ell$ connects.
By substituting equation \eqref{eq:angle-kvl} into equation \eqref{eq:cycle}, KVL can be expressed in terms of the power flows as
\begin{equation}
	\sum_{\ell} C_{\ell c}^0 x_\ell^0 f_\ell^0 = 0 \hspace{1cm} \forall c=1,\ldots,|\cLo|-|\cN|+1. \label{eq:cycle-kvl}
\end{equation}
Consequently, while the angle-based formulation defines KCL and KVL with $|\cLo|+|\cN|$ variables and $|\cLo|+|\cN|$ independent equality constraints,
the equivalent cycle-based formulation requires only $|\cLo|$ variables and $|\cLo|$ independent equality constraints.
Besides fewer variables and constraints, the cycle-based formulation moreover features sparser constraints than the angle-based formulation.

The computational appeal of this reformulation was evaluated in \cite{cycleflows} for multi-period linear optimal power flow problems
with generator capacity expansion and has been applied in other publications \cite{carvalho_cycleflows,dual_line_outages,dual_ptdf,transmission_switching,bollobas_graph_theory}.
It has further been proven in \cite{bollobas_graph_theory} the cycle-based formulation also holds for multigraphs\footnote{Multigraphs are graphs allowing parallel edges between the same two vertices.}
which is particularly relevant for its extension to transmission expansion planning.

\subsubsection{Post-facto Calculation of Voltage Angles}
\label{sec:postfacto}

The cycle-based formulation does not include variables for the voltage angles. However, if needed, they can be calculated subsequently using optimized net nodal power injection or consumption $p_i$.
By substituting equation \eqref{eq:angle-kvl} into equation \eqref{eq:kcl} one obtains
\begin{equation}
	p_i = \sum_{\ell\in\cLo} K_{i\ell} \frac{1}{x_\ell}\sum_{j\in \cN} K_{j\ell} \theta_j \qquad \forall i \in \cN.
\end{equation}
This can be rewritten with a weighted Laplacian $L = KBK^\top$
where $B$ is a diagonal matrix with $B_{\ell\ell}=b_\ell=x_\ell^{-1}$, leading to
\begin{equation}
	p_i = \sum_j L_{ij} \theta_j \qquad \forall i \in \cN.
\end{equation}
This can be solved for $\theta_i$ with
\begin{equation}
	\theta_i = \sum_j \left(L^{-1}\right)_{ij} p_j \qquad \forall i \in \cN.
	\label{eq:inverted}
\end{equation}
However, $L$ is not invertible as it has a zero eigenvalue with eigenvector $\bm{1}$.
Since equation \eqref{eq:slack} provides a reference voltage angle for one bus,
the remaining voltage angles $\{\theta_i\}_{i\in\cN\setminus\cN_0}$ can be found by inverting the
submatrix $L'\in\mathbb{R}^{|\cN\setminus\cN_0|\times|\cN\setminus\cN_0|}$ of $L$ which omits the row and column corresponding to the slack bus.

\section{Transmission Expansion Planning Formulations}
\label{sec:tep}

In transmission expansion planning (TEP) we consider the discrete reinforcement
of transmission lines based on a set of candidate lines $\cLi$.
The label $1$ indicates candidate lines.
We extend the optimization problem from Section \ref{sec:lopf} by
introducing a binary investment variable $i_\ell\in \mathbb{B}$
for each candidate line $\ell \in \cLi$ and then formulate constraints
on the power flow $f_\ell^1$ depending on the investment decision.

If the candidate line $\ell$ is not built, the power flow must be zero.
Otherwise, the absolute power flow must not exceed the nominal capacity $F_\ell^1$ of the candidate line.
This is expressed by the constraint
\begin{equation}
	|f_\ell^1| \leq i_\ell F_\ell^1 \qquad \forall \ell \in \cLi.
	\label{eq:linelimit}
\end{equation}
Just like existing lines $\ell \in \cLo$ the capital cost of candidate lines  $\ell \in \cLi$ are included in the objective function in equation \eqref{eq:objective}
and nodal balance constraints in equation \eqref{eq:kcl} defining Kirchhoff's Current Law (KCL).

To define Kirchhoff's Voltage Law (KVL) for candidate lines, both angle-based and cycle-based KVL constraints, given in equations \eqref{eq:angle-kvl} and \eqref{eq:cycle-kvl},
need to be edited such that they can only be active if the associated candidate lines are built.
To achieve this, both formulations make use of the big-$M$ disjunctive relaxation.
These modifications are subsequently developed in Section \ref{sec:angle-tep} for the angle-based and Section \ref{sec:cycle-tep} for the cycle-based KVL constraints.

\subsection{Angle-based Transmission Expansion Planning}
\label{sec:angle-tep}

The angle-based KVL constraint of the TEP problem is widely known from \cite{binato_benders,convex_optimization,gtep,romero}.
It transforms the KVL equality constraint from equation \eqref{eq:angle-kvl} into the two inequalities
\begin{align}
	f_\ell^1 - \frac{\theta_\ell}{x_\ell^1}
	 & \geq - M_\ell^{\text{KVL}} (1 - i_\ell)  \nonumber                     \\
	 & \leq + M_\ell^{\text{KVL}} (1 - i_\ell)  \qquad \forall \ell \in \cLi,
	\label{eq:angle-kvl-tep}
\end{align}
where, as previously, $\theta_\ell=\theta_i-\theta_j$.
If the big-$M$ parameters $M_\ell^{\text{KVL}}$ are suitably chosen, the inequalities in equation \eqref{eq:angle-kvl-tep} are inactive
if $i_\ell=0$, but together form the original equality constraint if $i_\ell=1$.

However, big-$M$ parameters are known to easily incur numerical challenges \cite{binato_benders,lumbreras_realsized}.
Therefore, $M_\ell^{\text{KVL}}$ are ideally chosen as large as necessary, to guarantee that the KVL constraint is inactive
whenever the candidate line is not built, and as small as possible, to avoid a detrimental large value range in the constraint matrix.

For the derivation of appropriate values it is necessary to distinguish candidate lines which connect buses within the same
synchronous zone ($\cLi_{\mathrm{intra}} \subseteq \cLi$, Section \ref{sec:angle-bigm-intra}) and candidate lines which connect
multiple synchronous zones ($\cLi_{\mathrm{inter}} \subseteq \cLi$, Section \ref{sec:angle-bigm-inter}).

\subsubsection{Big-$M$ Parameters for KVL Constraints Within Synchronous Zone}
\label{sec:angle-bigm-intra}

\begin{figure}
	Example A.1 \hspace{0.5cm} Example A.2 \hspace{0.9cm} Example A.3 \\
	\centering
	\includegraphics[width=0.2\columnwidth]{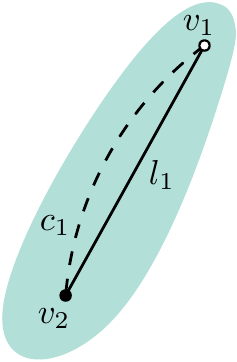}
	\includegraphics[width=0.33\columnwidth]{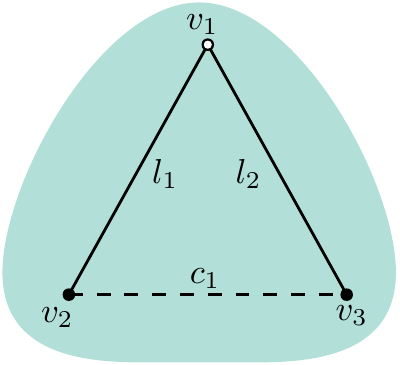}
	\includegraphics[width=0.33\columnwidth]{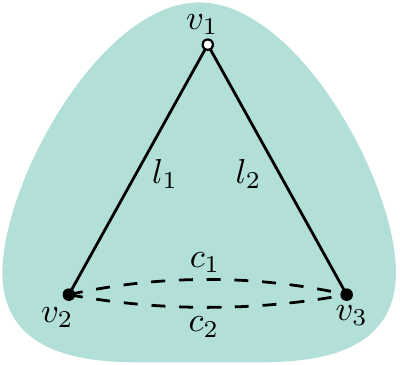}
	\caption{Example Group A. Candidate lines within a synchronous zone. Candidate lines denoted by $c_i$ and existing lines by $l_i$.}
	\label{fig:example-a}
\end{figure}

The derivation of minimal values for $M_\ell^{\text{KVL}}$ for candidate lines $\ell \in \cLi_{\mathrm{intra}}$
which connect buses of the same synchronous zone largely follows \cite{binato_benders, binato_phd},
but is reproduced here to facilitate a comparison with the novel cycle-based formulation and to set the notation.
\begin{theorem}
	The value of the disjunctive constant $M_{\ell}^\text{KVL}$ for a candidate line $\ell$ that connects two buses $i$ and $j$ of the same synchronous zone can be chosen following
	\begin{equation}
		M_{\ell}^\text{KVL} \geq \frac{|\sp|}{x_\ell^1}
	\end{equation}
	where $|\sp|$ is the length of the shortest path between the buses $i$ and $j$ along edges $k$ of the existing network graph $\mathcal{G}=(\cN, \cLo)$ with weights $F_{k}^0x_{k}^0$.
\end{theorem}
\begin{proof}
	Let $\ell\in\cLi_{\mathrm{intra}}$ be a particular candidate line for which equations \eqref{eq:linelimit} and \eqref{eq:angle-kvl-tep} hold.
	In the case $i_{\ell}= 0$ it follows from equation \eqref{eq:linelimit} that $f_{\ell}^1 = 0$ and from equation \eqref{eq:angle-kvl-tep} that
	\begin{equation}
		- M_{\ell}^\text{KVL}x_{\ell}^1 \leq \theta_i - \theta_j \leq M_{\ell}^\text{KVL}x_{\ell}^1.
		\label{eq:difflimit}
	\end{equation}
	Equation \eqref{eq:difflimit} represents a limit on the voltage angle difference and the value of $M_{\ell}^\text{KVL}$
	must be chosen such that for as long as $i_{\ell}= 0$ this limit is never reached.
	Otherwise invalid limits on the angle difference are imposed.
	We must therefore derive valid big-$M$ parameters from constraints on the voltage angle difference that are already enforced through the existing network.

	If there exists a line $\ell\in \cLo$ in parallel to the candidate line
	(e.g. as in Example A.1 in Figure \ref{fig:example-a})
	we can obtain these by substituting equation \eqref{eq:angle-kvl} into equation \eqref{eq:existinglinelimit}, yielding the limits
	\begin{equation}
		- F_{\ell}^0x_{\ell}^0 \leq \theta_i - \theta_j \leq F_{\ell}^0x_{\ell}^0.
		\label{eq:parallelangleconstraint}
	\end{equation}
	By combining equations \eqref{eq:difflimit} and \eqref{eq:parallelangleconstraint} we can retrieve a minimum value for $M_{\ell}^\text{KVL}$:
	\begin{equation}
		M_{\ell}^\text{KVL} \geq \frac{F_{\ell}^0x_{\ell}^0}{x_{\ell}^1}
	\end{equation}
	Now consider the slightly more complicated case where the candidate line $\ell$ is not a duplication of an existing line (e.g. as in Example A.2 in Figure \ref{fig:example-a}).
	The theorem specifies that the buses $i$ and $j$ of $\ell$ are part of the same synchronous zone.
	Thus, there is at least one sequence $\mathcal{P}_{i,j}=\{k(i,b_1),k(b_1,b_2),\dots,k(b_n,j)\}$ of existing lines $k\in\cLo$
	along buses $\{b_m\}_{m=1,\dots,n}$ which already connects these buses.
	Hence, just as with an existing parallel line there is an existing limit on the voltage angle difference,
	only that the limit is not given by just one existing line but by a set of existing lines:
	\begin{equation}
		- \sum_{k\in\mathcal{P}_{i,j}} F_{k}^0x_{k}^0 \leq \theta_i - \theta_j \leq \sum_{k\in\mathcal{P}_{i,j}} F_{k}^0x_{k}^0
		\label{eq:maxanglediff}
	\end{equation}
	To find the tightest limit on $\theta_i - \theta_j$ we need to find the shortest path $\sp$ among all possible paths $\mathcal{P}_{i,j}$
	with weights $F_{k}^0x_{k}^0$ using e.g. the Dijkstra algorithm, which then yields
	\begin{equation}
		M_{\ell}^\text{KVL} \geq \frac{|\sp|}{x_\ell^1} = \frac{\sum_{k\in\sp}F_k^0x_k^0}{x_\ell^1}
	\end{equation}
	as specified in the theorem.
\end{proof}

\subsubsection{Big-$M$ Parameters for KVL Constraints Across Synchronous Zones}
\label{sec:angle-bigm-inter}

If the buses connected by candidate line $\ell$ are not part of the same synchronous zone and therefore no path exists to
infer an existing limit on the voltage angle difference, it is possible to fall back to a significantly larger value
\begin{equation}
	M_\ell^\text{KVL} \geq \frac{\sum_{k \in \cLo \cup \cLi} F_k x_k}{x_\ell^1}
\end{equation}
which has been proven in \cite{tsamasphyrou} to be a valid choice for any combination of
line investment decisions, under the condition that a reference angle is defined for all
synchronous zones such that $\theta_i = 0 \;\forall i\in\cN_0$ if no new lines are built.
Otherwise, due to the rotational degeneracy no relation could be established between
the nodal voltage angles of disconnected networks.

\subsubsection{Slack Constraints Across Synchronous Zones}
\label{sec:angle-bigm-slack}

If multiple synchronous zones may be connected by building new lines, the slack constraint in equation \eqref{eq:slack} must also be modified, since it applies separately in each connected network. When two networks are connected, one of the slack constraints should be relaxed.
The slack constraints cannot simply be dropped because the derivation of big-$M$ parameters for the KVL constraints
across synchronous zones (Section \ref{sec:angle-bigm-slack})
depends on a calculable maximal voltage angle difference across
synchronous zones even if they are not coupled.
Available transmission expansion studies that alleviate rotational degeneracy of voltage angles with slack constraints have not dealt with this case. In this section
a novel treatment of the connection of multiple synchronous zones is provided that handles the slack constraints by managing the combinatorics of possible relaxations that apply as networks are connected.

\begin{figure}
	\flushleft
	Example C.1 \\
	\flushright
	\includegraphics[width=0.65\columnwidth]{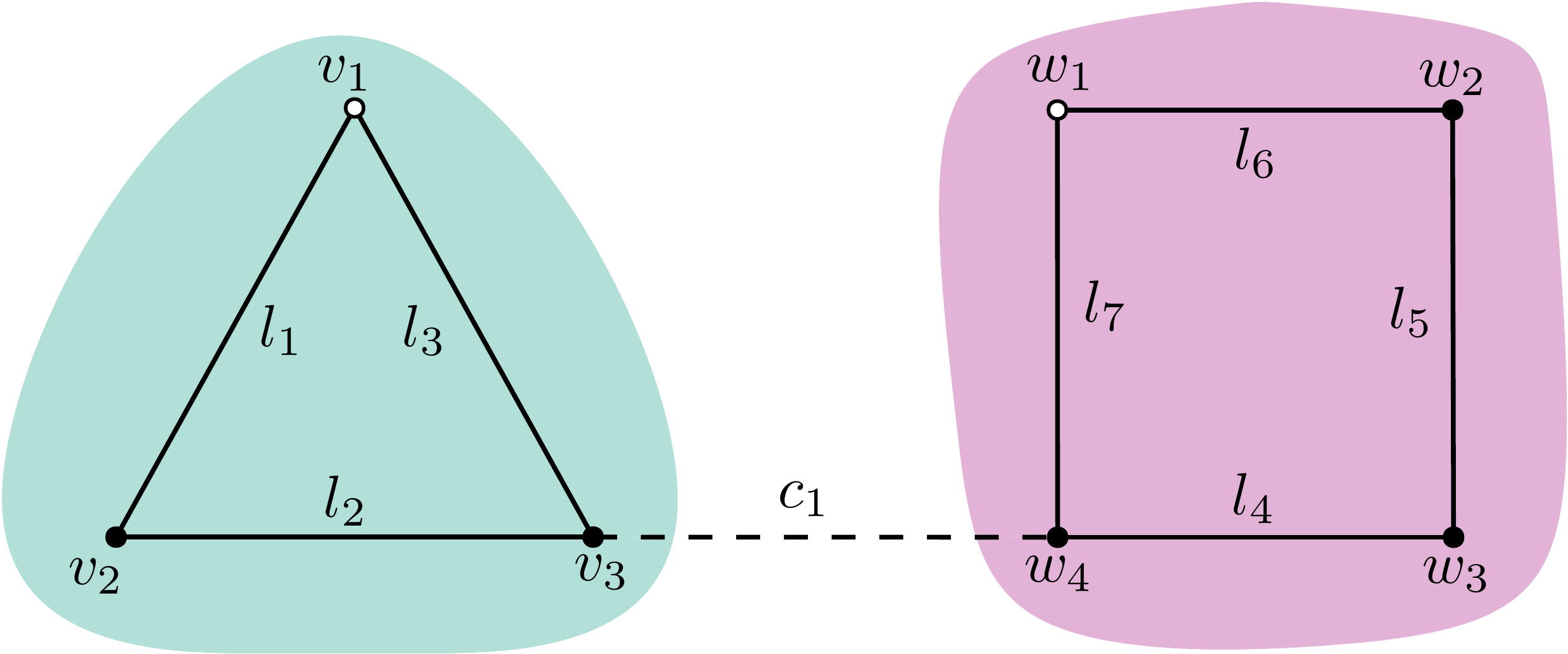}\\
	\flushleft
	Example C.2 \\
	\flushright
	\includegraphics[width=0.65\columnwidth]{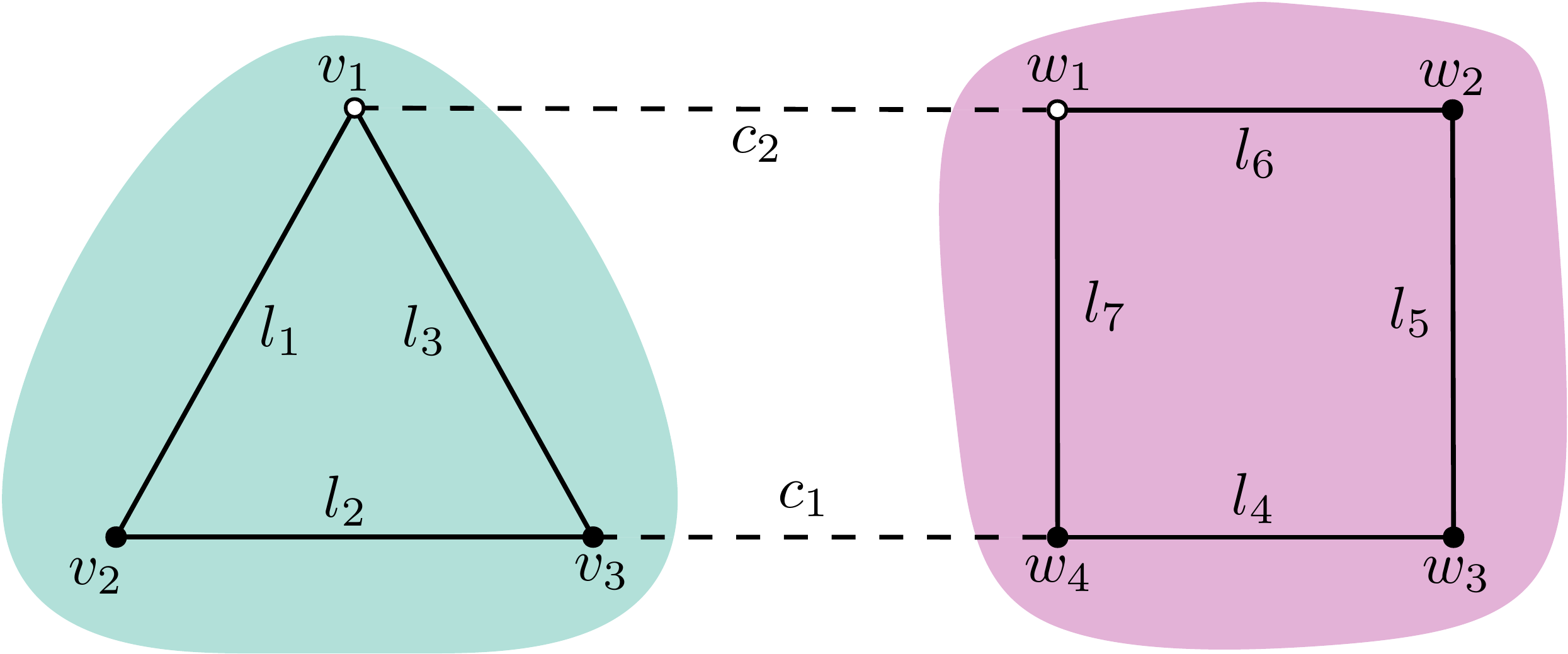}\\
	\flushleft
	Example C.3 \\
	\flushright
	\includegraphics[width=1\columnwidth]{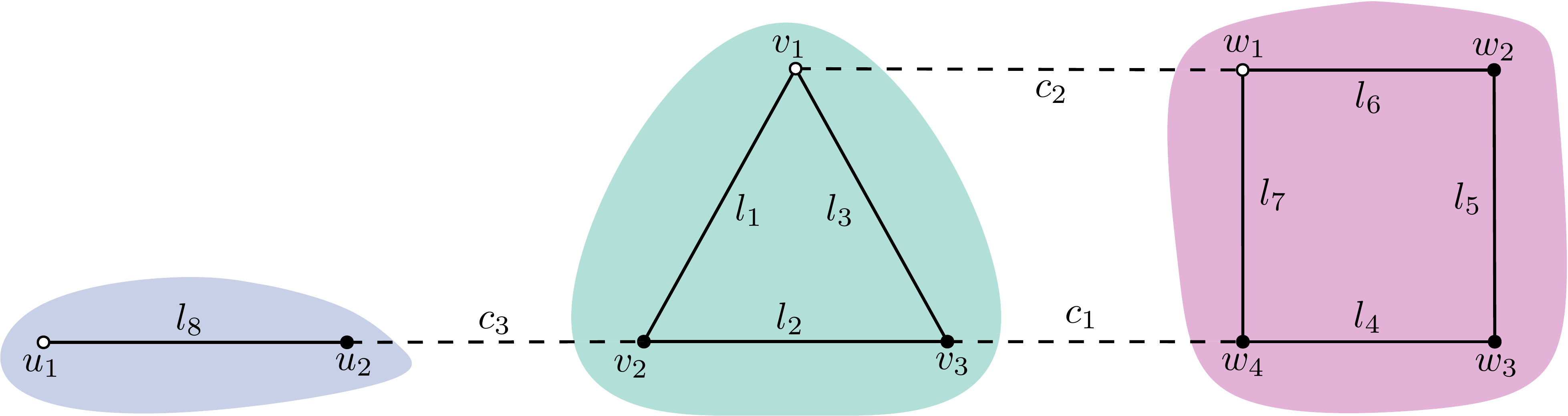}
	\caption{Example Group C. Candidate lines across synchronous zones. Candidate lines denoted by $c_i$ and existing lines denoted by $l_i$.}
	\label{fig:example-c}
\end{figure}

\begin{figure}
	remote root / depth-first \hspace{1.1cm} central root / breadth-first
	\centering
	\includegraphics[width=0.49\columnwidth]{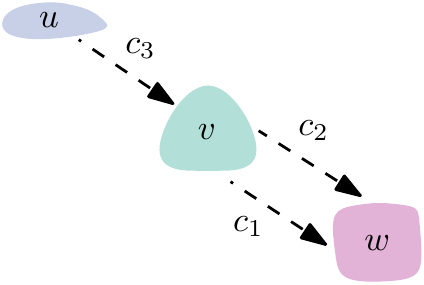}
	\includegraphics[width=0.49\columnwidth]{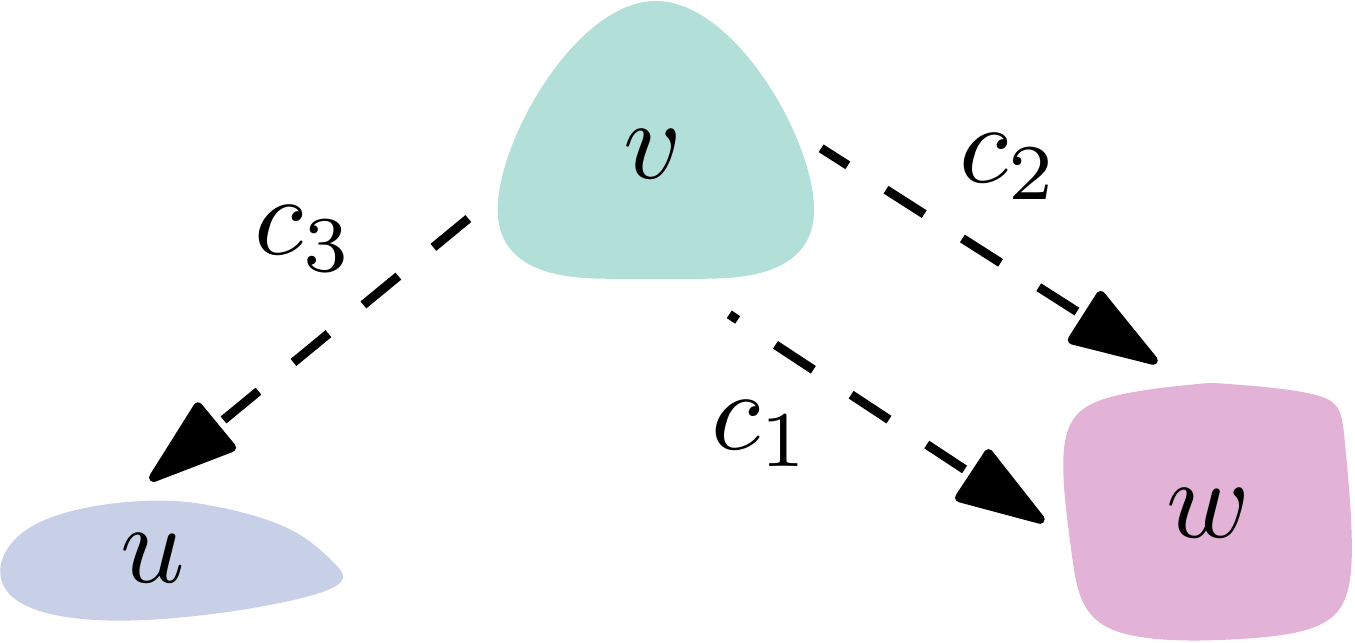}
	\caption{Example C.3 shown as different directed rooted trees of the subnetwork graph. In the depth-first variant, $u$ is the root subnetwork, $c_3$ relaxes the slack of $v$ and $c_1$ or $c_2$ relax the slack of $w$. In the breadth-first variant, $v$ is the root subnetwork, $c_3$ relaxes the slack  of  $u$ and $c_1$ or $c_2$ relax the slack of $w$. }
	\label{fig:example-c-sub}
\end{figure}

Initially, consider Example C.1 in Figure \ref{fig:example-c} where $c_1$ is a candidate line which, if built, would synchronize two synchronous zones $v$ and $w$.
If $c_1$ is built, one of the constraints in equation \eqref{eq:slack} regarding the two slack buses $v_1$ and $w_1$ must be rendered ineffective.
Otherwise the nodal voltage angles would be fixed at two buses within the same synchronous zone,
but the flow is determined by the voltage angle difference between buses. The solution would yield invalid or infeasible power flows.
Therefore, we adjust the slack constraint of $w$ to $|\theta_{w_1}| \leq i_{c_1} M_{c_1}^{\text{slack}}$, where $M_{c_1}^{\text{slack}}$ is a sufficiently large constant.

Now consider Example C.2 in Figure \ref{fig:example-c} where additionally $c_2$ is a candidate line which connects
the same two synchronous zones as $c_1$. In this case, we must agree on a single slack constraint relaxed by $c_1$
and $c_2$ as otherwise, if both are built, no slack constraint would remain to alleviate rotational degeneracy.
Hence, the slack constraint of $w$ is adjusted to $|\theta_{w_1}| \leq \sum_{\ell\in\{c_1, c_2\}} i_{\ell} M_{\ell}^{\text{slack}}$.
The sum on the right-hand side acts as a logical OR expression such that each positive investment decision $i_\ell$ alone renders the constraint non-binding.

Next, consider the slightly more complicated Example C.3 in Figure \ref{fig:example-c} where three synchronous zones may be
synchronized by candidates $c_1$, $c_2$ and $c_3$.
In this case, it is essential to select a single root synchronous zone,
the slack constraint of which is to be kept if all candidate lines are built.
For instance, not all three candidate lines can relax the slack constraint of $v$
as this would result in two remaining slack constraints in one synchronous zone.

Figure \ref{fig:example-c-sub} sketches two possible relations between the candidate lines and the slack constraints they relax
without the need to consider complementary investment decisions.
It shows reduced graphs where
the nodes $\mathcal{S}$ represent all synchronous zones and
the directed edges represent the candidate lines in $\cLi_{\text{inter}}$ and point to the synchronous zone they affect.
Since the connecting nodes are formally different than in $\mathcal{G}$ we label this edge set with $\cLi_{\mathcal{S}}$.
In the following, we refer to this graph as the subnetwork graph $\mathcal{G}_{\mathcal{S}} = (\mathcal{S}, \cLi_{\mathcal{S}})$.

Generalizing from the examples, we define sets of candidate lines $\cLi_v\subseteq\cLi_{\mathrm{inter}}$ which should turn the slack constraint of synchronous zone $v$ non-binding.
We can achieve a structure without complicating interdependencies of line investment variables if the graph of subnetworks
$\mathcal{G}_{\mathcal{S}}$ is a forest of directed trees with a defined root (but allowing parallel edges).
With an associated big-$M$ constant $M_\ell^{\text{slack}}$ that is large enough regardless of all other investment decisions
(cf. Section \ref{sec:angle-bigm-slack}), we reformulate the slack constraints to
\begin{equation}
	|\theta_v| \leq \sum_{\ell \in \cLi_v} i_\ell M_\ell^{\text{slack}}
	\label{eq:bigm-slack}
\end{equation}
which are correct for any combination of line investments.

If the subnetwork graph would not be a forest of directed rooted trees (with parallel edges),
more interdependencies would arise due to the manifold of combinations of synchronization scenarios.
Consider Example D.1 in Figure \ref{fig:example-d} where considering a dependency is inevitable.
It is viable to encode one logical AND expression for two binary investment variables $i_1$ and $i_2$ in linear programming with
an auxiliary variable $i_{12}$ and the constraint
\begin{equation}
	0\leq i_{1} + i_{2} - 2 i_{12} \leq 0
\end{equation}
\cite{logical_and}.
But the rapidly growing number of additional binary auxiliary variables and constraints that would be required for
only marginally more complicated cases, such as Example D.2, add to the appeal of reformulating the problem
without voltage angle variables in cases where multiple synchronous zones may be joined.

\subsubsection{Big-$M$ Parameters for Slack Constraints}

Having established that the subnetwork graph $\mathcal{G}_{\mathcal{S}}$ must be a forest of directed rooted trees
in order to avoid considering interdependencies of investments,
this section derives suitable big-$M$ parameters for the modified slack constraints in equation \eqref{eq:bigm-slack}.
It follows a similar logic as the preceding derivation for the KVL constraints in Section \ref{sec:angle-bigm-intra}.

For a start consider the simple case where there is only a single
candidate line $\ell$ that would connect two asynchronous zones with reference buses $v_1$ and $w_1$.
Choose, without loss of generality, that $\ell$ relaxes the slack constraint of $v_1$ ($\ell\in\cLi_{v_1}$).
Then if the candidate line built ($i_\ell=1$),
\begin{equation}
	\theta_{w_1} = 0 \qquad \text{and} \qquad |\theta_{v_1}|\leq M_\ell^\text{slack},
\end{equation}
where $M_\ell^\text{slack}$ is chosen such that the constraint is never binding.
To determine $M_\ell^\text{slack}$ we need to find the maximum absolute voltage angle $|\theta_{v_1}|$
if the candidate line is built.
This depends on the reference voltage angle $\theta_{w_1}$.
We can relate $\theta_{v_1}$ and $\theta_{w_1}$ by following a path $\mathcal{P}^\ell_{v_1,w_1}$ between the
slack buses $v_1$ and $w_1$ through the graph $\mathcal{G}^\ell=(\mathcal{N}, \cLo\cup\{\ell\})$
that consists of the existing network plus the candidate line $\ell$ via
\begin{equation}
	\theta_{v_1} - \theta_{w_1} = \sum_{ij\in \mathcal{P}_{v_1,w_1}^\ell} \theta_i - \theta_j.
\end{equation}
One can easily see this by following Example C.1 in Figure \ref{fig:example-c}.
\begin{equation}
	(\theta_{v_1}-\theta_{v_3}) + (\theta_{v_3}-\theta_{w_4}) +(\theta_{w_4}-\theta_{w_1}) =\theta_{v_1} - \theta_{w_1}.
\end{equation}
Knowing this we can calculate the maximum voltage angle difference between the two slack buses,
as previously done in equation \eqref{eq:maxanglediff} using the shortest path along lines in
$\mathcal{G^\ell}$ with weights $F_\ell x_\ell$ to determine a lower bound for $M_\ell^\text{slack}$:
\begin{equation}
	M_\ell^\text{slack} \geq \sum_{k\in \mathcal{P}^{\ell,\text{min}}_{v_1,w_1}} F_kx_k.
\end{equation}
Now consider the slightly more complicated case of $|\cLi_{v_1}|\geq 2$ candidate lines $\ell \in \cLi_{v_1}$
where either line potentially synchronizes two separate power networks with reference buses $v_1$ and $w_1$.
We can repeat the preceding calculation of
$M_\ell^\text{slack}$ for each candidate line $\ell\in \cLi_{v_1}$.
However, the maximum voltage angle difference irrespective of all investment combinations is $\max\left\{M_\ell^\text{slack}\;\vert\;l\in\cLi_{v_1}\right\}$
and should therefore be chosen for both lines.

A hierarchical strategy based subnetwork graph $\mathcal{G}_\mathcal{S}$ is applied if multiple synchronous zones can be connected.
We add the maximum big-$M$ parameter of the upstream synchronous zone to all big-$M$ parameters of the downstream synchronous zones,
starting at the root.
For instance, in Example C.3 in Figure \ref{fig:example-c-sub} using the remote root variant, the big-$M$ constant for $c_3$ would be added to those of $c_1$ and $c_2$.
This approach does not yield minimal values, as it takes a detour via the slack bus of intermediate synchronous zones,
but circumvents the need to consider investment dependencies to guarantee non-binding slack constraints.
Due to this hierarchical approach, choosing a tree via breadth-first search from a central node of the subnetwork graph $\mathcal{G}_\mathcal{S}$
is advantageous as it generally results in lower big-$M$ constants.

\subsection{Cycle-based Transmission Expansion Planning}
\label{sec:cycle-tep}

Investing in candidate lines in the transmission system can incur new cycles for which the KVL constraint in equation \eqref{eq:cycle-kvl}
must hold if and only if all candidate lines which are part of a new cycle are built.
In the following these will be referred to as candidate cycles.
Both existing and candidate lines can be involved in a candidate cycle.
Given these candidate cycles as an incidence matrix $C_{\ell c}^1$ where $\ell \in \cLo\;\cup\;\cLi$
we can formulate the KVL constraints analogously to the cycle-based load flow formulation from equation \eqref{eq:cycle-kvl}
such that it is enforced only if all candidate lines of that cycle are built:
\begin{align}
	\sum_{\ell\in\cLo\;\cup\;\cLi} C_{\ell c}^1 x_\ell f_\ell & \geq -  M_c^\text{KVL}  \left(\sum_{\ell \in \cLi} C_{\ell c}^1 (1-i_\ell)\right) \nonumber            \\
	                                                          & \leq + M_c^\text{KVL}  \left(\sum_{\ell \in \cLi} C_{\ell c}^1 (1-i_\ell)\right) \qquad  \forall \; c.
	\label{eq:cycle-tep-kvl}
\end{align}
Like in the angle-based TEP formulation (cf.\ Section \ref{sec:angle-tep}),
the cycle-based TEP formulation relies on the big-$M$ disjunctive relaxation with a
sufficiently large parameter $M_c^\text{KVL}$ for each candidate cycle $c$.
The candidate cycle matrix $C_{\ell c}^1$ on the right-hand side acts as an
indicator for whether candidate line $\ell$ is contained within the candidate cycle $c$.
Only if all those $i_\ell=1$, equation \eqref{eq:cycle-tep-kvl} becomes binding.

The cycle-based linear power flow equations have previously been applied to
the related optimal transmission switching (OTS) \cite{transmission_switching}.
However, using cycle-based power flow constraints in TEP has a distinct advantage over using it in OTS.
Since usually in TEP problems existing transmission infrastructure cannot be removed,
the KVL constraints from equation \eqref{eq:cycle-kvl} remain valid, regardless of the binary decision variables.
Conversely, OTS needs to consider all simple cycles from the start because the initial network topology, and therefore the cycle basis, may not persist \cite{transmission_switching}.
For TEP it is enough to append KVL constraints for supplemental candidate cycles according to equation \eqref{eq:cycle-tep-kvl}.

Candidate cycles can originate from
(i) a candidate line parallel to an existing line,
(ii) a candidate line connecting two buses which are already connected and are thereby part of the same synchronous zone, or
(iii) multiple candidate lines connecting two or more synchronous zones which form cycles in the subnetwork graph $\mathcal{G}_\mathcal{S}$.

\subsubsection{Candidate Cycles Within Synchronous Zone}

\begin{figure}
	\centering
	Example B.1 \hspace{1.9cm} Example B.2 \\
	\includegraphics[width=0.46\columnwidth]{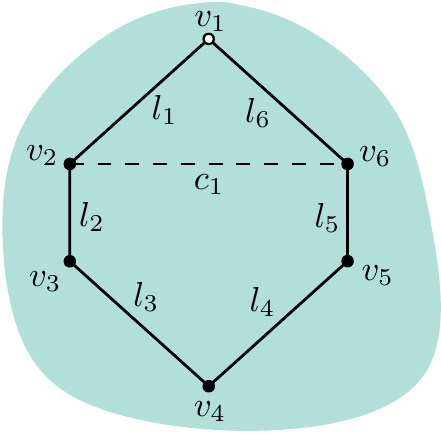}
	\includegraphics[width=0.46\columnwidth]{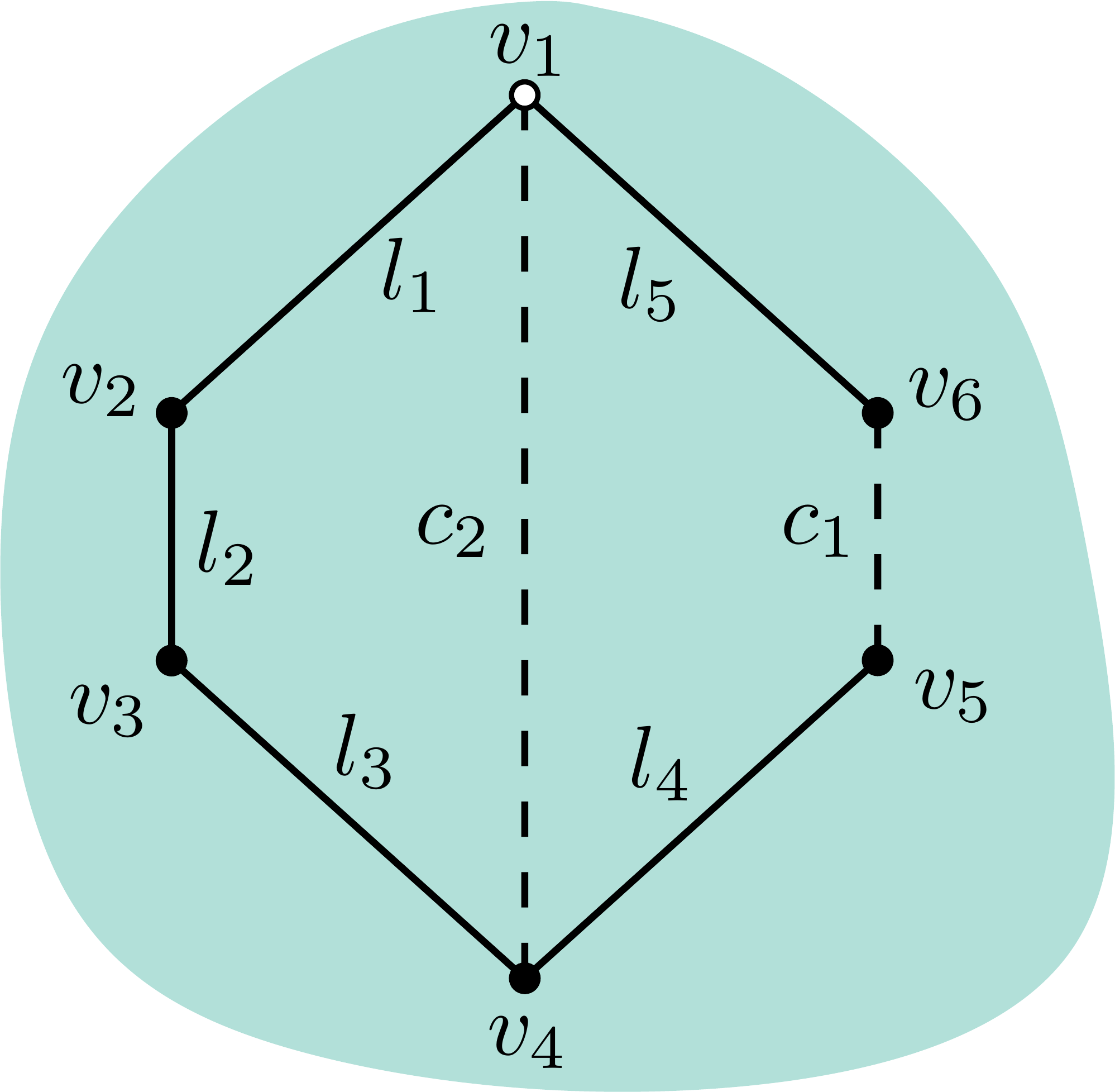}
	\caption{Example Group B. Choice of candidate cycles within synchronous zone. Candidate lines denoted by $c_i$ and existing lines denoted by $l_i$.}
	\label{fig:example-b}
\end{figure}

Finding candidate cycles within the same synchronous zone follows the subsequently described algorithm:
For each candidate line $\ell\in\cLi$ connecting buses $i$ and $j$ find a shortest path $\sp$ through the network graph $\mathcal{G}=(\cN, \cLo)$ with edge weights 1,
which includes only the existing transmission infrastructure.
The edges of the shortest path and the respective candidate line form a candidate cycle.
The cycle incidence vector is formed according to equation \ref{eqn:cycle-edge-matrix}.

While any path through $\mathcal{G}$ from $i$ to $j$ would yield a valid candidate cycle, it is computationally advantageous
to minimize the size of the cycles to obtain sparser KVL constraints.
For instance, in Example B.1 in Figure \ref{fig:example-b} the cycle for candidate line $c_1$ would consist
of $\{c_1, l_1, l_6\}$ and not $\{c_1, l_5, l_4, l_3, l_2\}$.
The potential KVL constraint would contain only three flow variables rather than five.

It is not required to add both cycles to the set of candidate cycles.
If $c_1$ gets built, already one cycle in addition to the initial cycle basis ($\{l_1, l_2, l_3, l_4, l_5, l_6\}$)
forms a cycle basis of the new network topology.

It is furthermore necessary to only consider existing lines and no other candidate lines for the shortest path search.
Otherwise a KVL constraint might be enforced only once a combination of candidate lines is built,
although building one of the candidate lines alone would already introduce a new cycle.
This is illustrated in Example B.2 in Figure \ref{fig:example-b}.
The cycles $\{c_1, l_5, c_2, l_4\}$ and $\{c2, l_1, l_2, l_3\}$ would incur incorrect KVL constraints if $c_1$ is built but not $c_2$.
On the contrary, the longer cycles $\{c2, l_1, l_2, l_3\}$ and $\{c1, l_5, l_1, l_2, l_3, l_4\}$ obtained through the cycle search
algorithm entail a correct modified cycle basis for either combination of investments.

\subsubsection{Candidate Cycles Across Synchronous Zones}

\begin{figure}
	\flushleft
	Example D.1
	\centering
	\includegraphics[width=0.9\columnwidth]{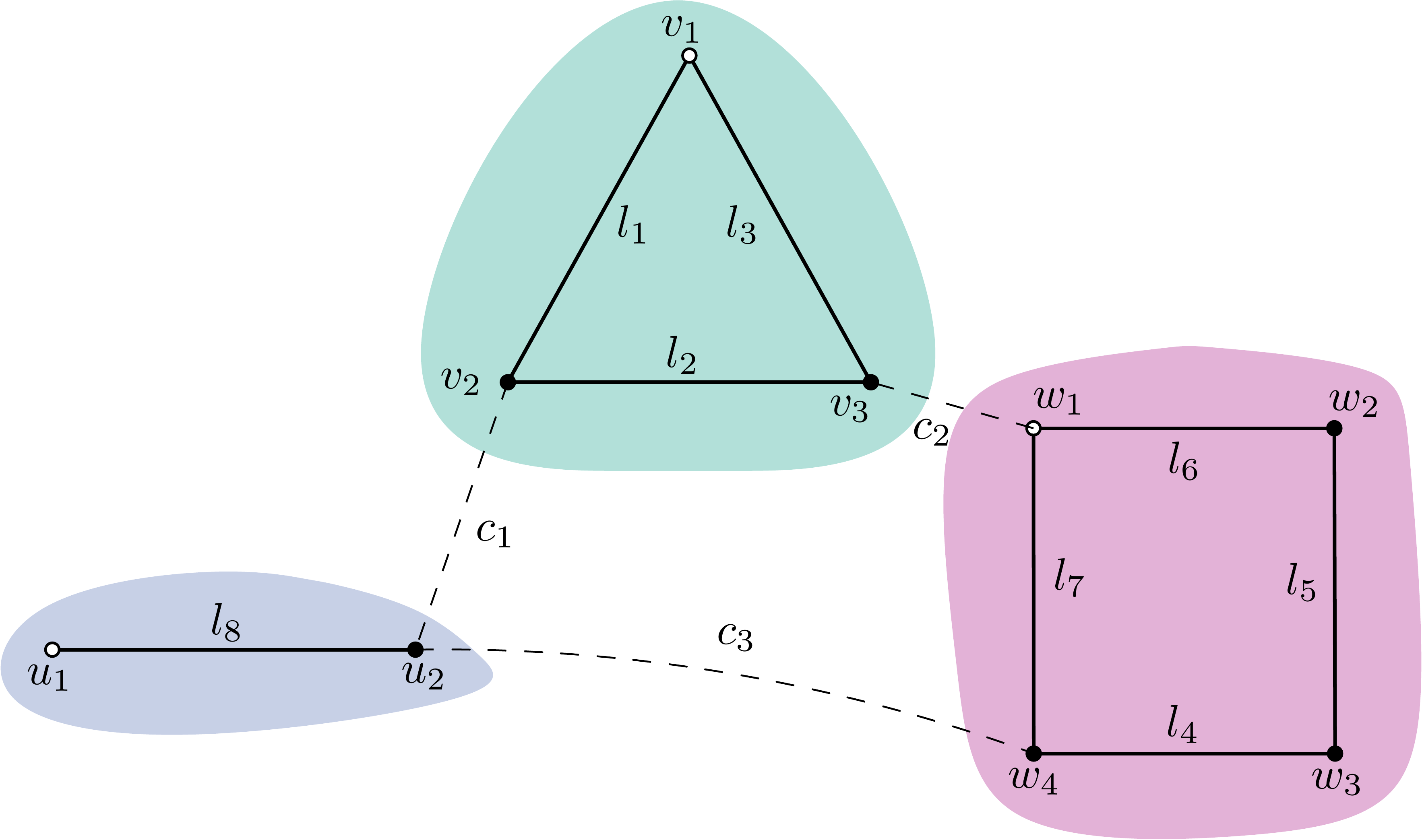}
	\flushleft
	Example D.2
	\centering
	\includegraphics[width=0.9\columnwidth]{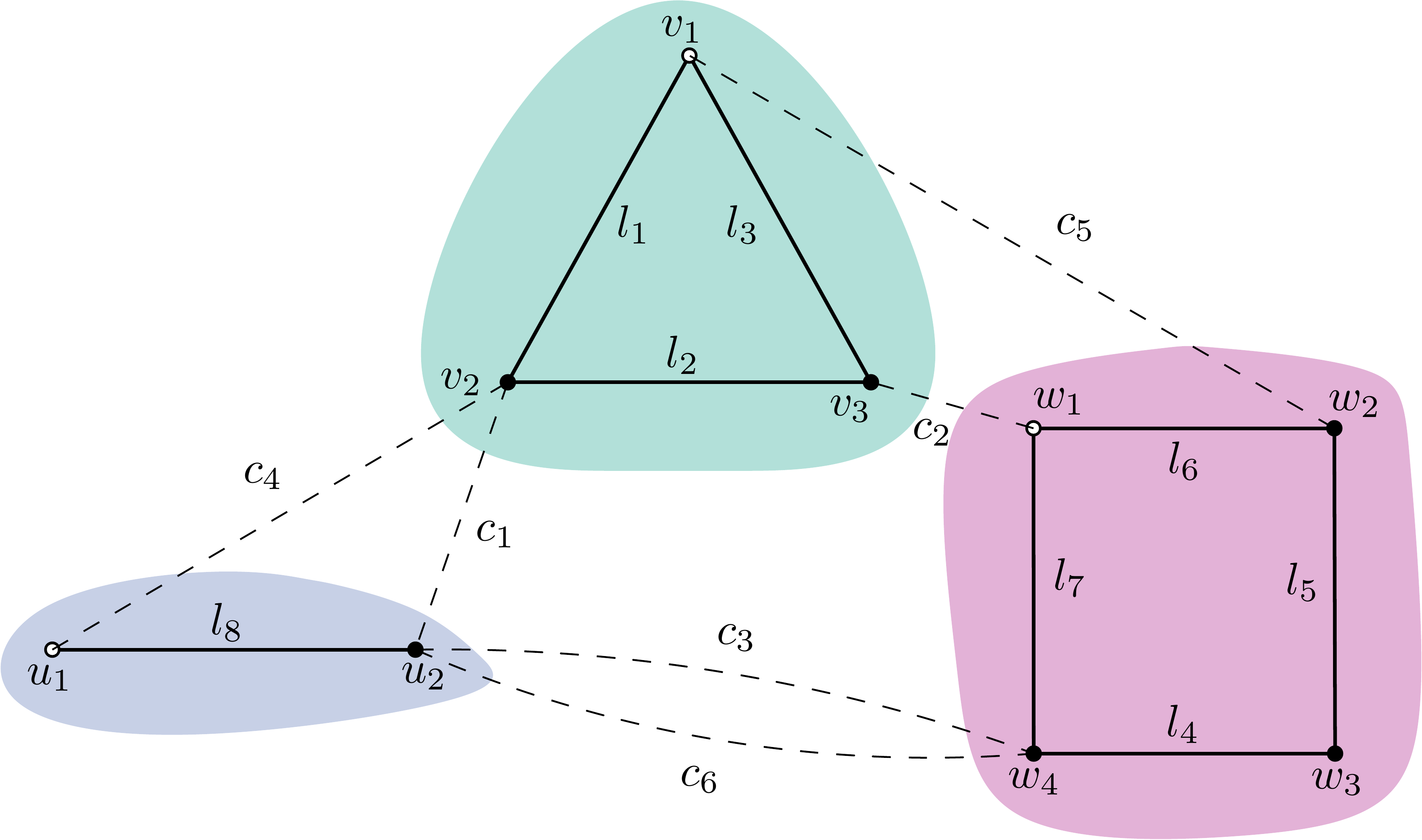}
	\caption{Example Group D. Choice of candidate cycles across synchronous zones and limits of the angle-based formulation. Candidate lines denoted by $c_i$ and existing lines denoted by $l_i$.}
	\label{fig:example-d}
\end{figure}

\begin{figure}
	Example D.1 \hspace{2.1cm} Example D.2
	\centering
	\includegraphics[width=0.49\columnwidth]{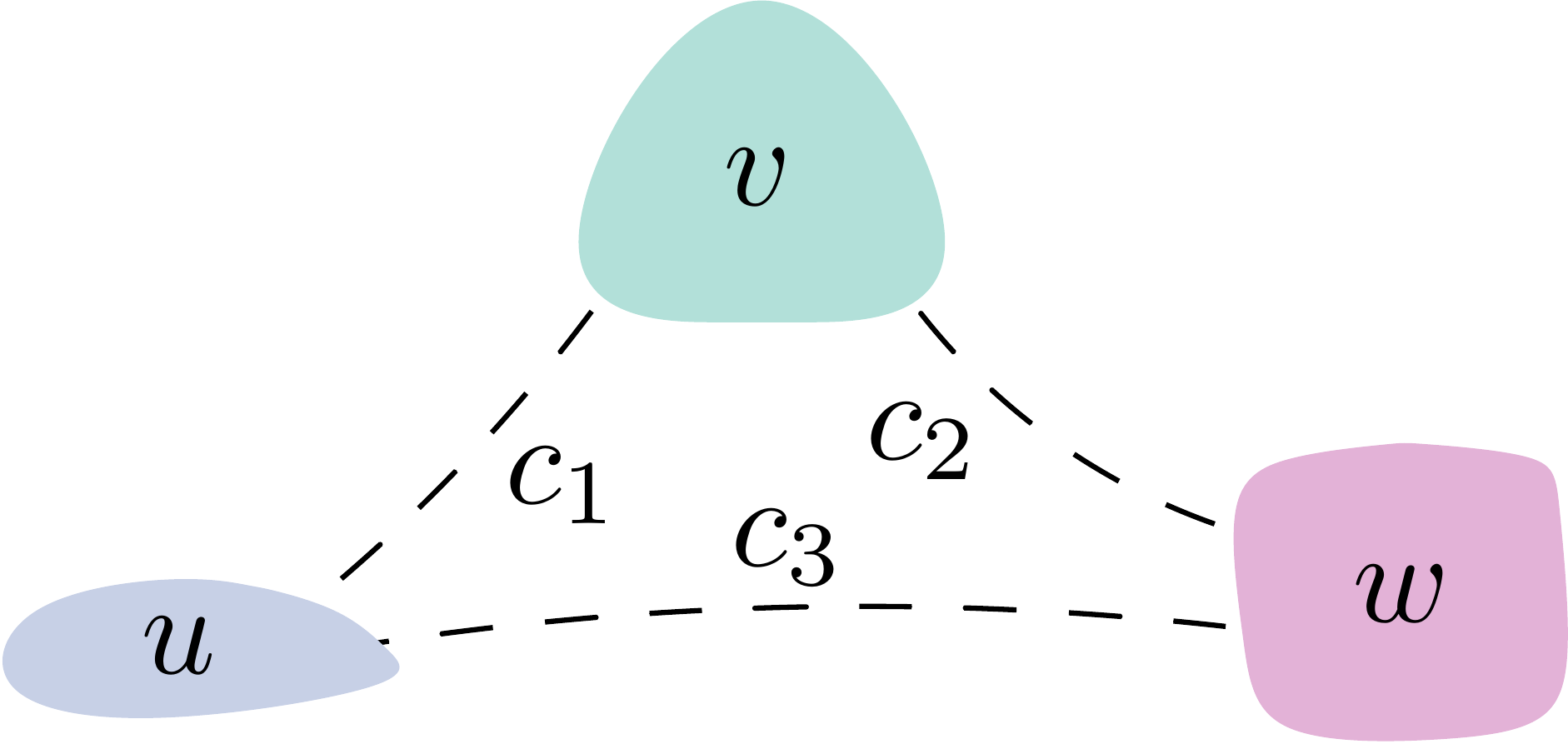}
	\includegraphics[width=0.49\columnwidth]{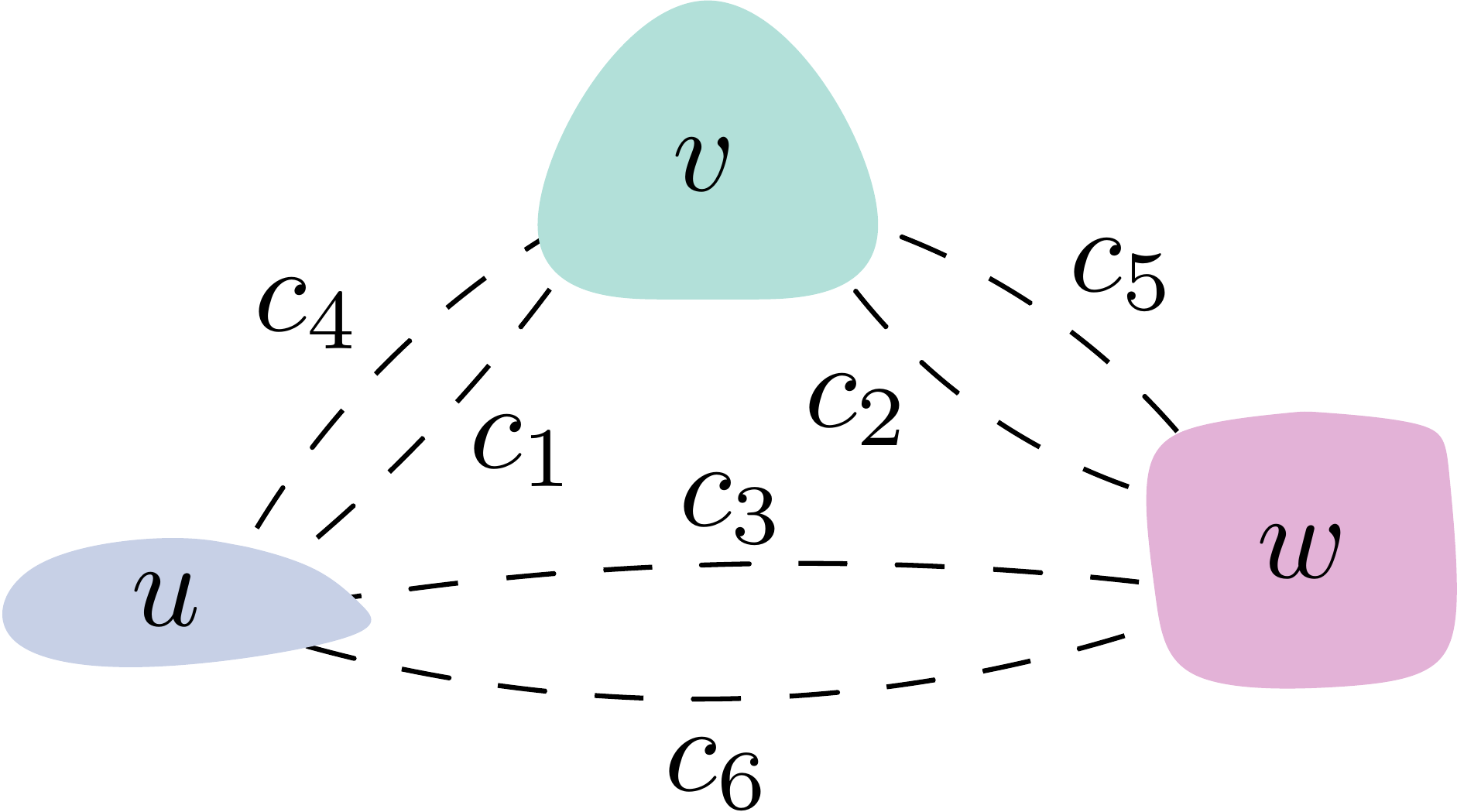}
	\caption{Example Group D as subnetwork graphs.}
	\label{fig:example-d-sub}
\end{figure}

If two synchronous zones can only be synchronized by one particular candidate line (cf. Example C.1 in Figure \ref{fig:example-c}),
no new cycle has to be added. Then KCL alone already determines the power flow.

A new cycle must be introduced if two candidate lines connect to the same two synchronous zones.
The cycle incidence vector is built from the two candidate lines and the existing lines on the
shortest paths of $\mathcal{G}$ through the synchronous zones between the connection points, where edge weights are set to $1$.
In Example C.2 in Figure \ref{fig:example-c}, $\{c_1, l_3,c_2,l_7\}$ would form the according candidate cycle.
Note, that also $\{c_1, l_2, l_3, c_2, l_6, l_5, l_4\}$ would be a correct candidate cycle, but the resulting conditional KVL constraint would be less sparse.

Additional cycles cannot only be incurred by the complementary investment of two candidate lines,
but also from multiple candidate lines connecting three or more synchronous zones
as depicted in Examples D.1 and D.2 in Figure \ref{fig:example-d}.
While Example D.1 has just one candidate cycle ($\{c_1, l_2, c_2, l_7, c_3\}$),
Example D.2 with two candidate lines per pair of synchronous zones already has 11 candidate cycles to consider (3 cycles with two edges and 8 cycles with three edges).
This is due to a growing number of interdependent combinations of investment decisions that would each demand conditional KVL constraints.
Example D.2 creates a similar situation as in the OTS problem \cite{transmission_switching},
where it becomes necessary to consider all simple cycles of the subnetwork graph $\mathcal{G}_{\mathcal{S}}$
(plus the corresponding shortest paths within the synchronous zones) a candidate cycle.
Nonetheless, the initial cycle basis of the network graph $\mathcal{G}$ still remains intact.

\subsubsection{Big-$M$ Parameters for KVL Constraints}

Having built the incidence matrix of the candidate cycles $C_{\ell c}^1$, the subsequent step is to derive an appropriate big-$M$ parameter $M_c^\text{KVL}$ for each candidate cycle.

\begin{theorem}
	The value of the disjunctive constant $M_c^\text{KVL}$ for a candidate cycle $c$ can be chosen following
	\begin{equation}
		M_c^\text{KVL} \geq \sum_{\ell\in\cLo\;\cup\;\cLi} C_{\ell c}^1 x_\ell F_\ell
		\label{eq:cycle-bigm}
	\end{equation}
\end{theorem}

\begin{proof}
	Let $c$ be a candidate cycle.
	If not all candidate lines of the candidate cycle are build, the corresponding cycle constraint must be inactive in all circumstances. In the case where $n$ lines are not built equation \eqref{eq:cycle-tep-kvl} becomes
	\begin{equation}
		- n M_c^\text{KVL} \leq \sum_{\ell\in\cLo\;\cup\;\cLi} C_{\ell c}^1 x_\ell f_\ell \leq n M_c^\text{KVL}.
		\label{eq:ccreduced}
	\end{equation}
	Moreover, through equation \eqref{eq:existinglinelimit} the flow $f_\ell$ in lines $\ell \in \cLo\;\cup\;\cLi$ is symmetrically limited by their nominal capacity $F_\ell$. Hence,
	\begin{equation}
		\sum_{\ell\in\cLo\;\cup\;\cLi} C_{\ell c}^1 x_\ell f_\ell \leq \sum_{\ell\in\cLo\;\cup\;\cLi} C_{\ell c}^1 x_\ell F_\ell.
		\label{eq:cclimit}
	\end{equation}
	Constraint \eqref{eq:ccreduced} must be inactive even if an investment decision for only one candidate line is missing to close the candidate cycle. Therefore with $n=1$ and the upper limit given in equation \eqref{eq:cclimit}, one obtains
	equation \eqref{eq:cycle-bigm}.
\end{proof}

Since there are no voltage angle variables and therefore no slack constraints in the cycle-based
formulation, there is no need to calculate such big-$M$ parameters.
For calculating the voltage angles as outlined in Section \ref{sec:postfacto},
the slack buses can simply be chosen based on the resulting synchronous zones
after the optimal investment decisions are known.
This has the advantage over the angle-based formulation that matters of synchronization
do not have to be encoded into the optimization problem.

\section{Experimental Setup}
\label{sec:setup}

We benchmark the presented transmission expansion planning formulations on multiple networks
using the open European power transmission system model PyPSA-Eur \cite{pypsaeur} as a basis, which
includes regionally resolved time series for electricity demand and renewable generator availability.
The evaluation criteria are computational speed and peak memory consumption.
The benchmark problems consider simultaneous generation and transmission
capacity expansion each given a carbon budget of 40 Mt$_{\text{CO}_2}$, following the
description of the long-term investment planning problem outlined in equation \eqref{eq:objective}.
Considered generation technologies include solar photovoltaics, onshore and offshore wind generators
as well as open-cycle gas turbines (OCGT) and run-of-river power plants,
but no storage units to maintain the independence of hourly snapshots and
focus on transmission expansion as balancing option for renewables.
For candidate lines we assume a standard line type for transmission lines
at 380 kV with a capacity of approximately 1.7 GW \cite{pypsa}.
Full model details and underlying assumptions are provided via the links provided in the appendix and in \cite{pypsaeur}.

To obtain a comprehensive sample of network topologies and operating conditions, we vary
the number of clustered nodes in Europe $\{1000,750,500,250\}$,
the number of selected hours from a whole year $\{1,5,25,50,75\}$,
the regional extract (see colored areas in Figure \ref{fig:geography}),
the tolerated MIP optimality gap $\{0.5\%, 1\%\}$, and
the number of candidate lines per existing HVAC and HVDC corridor $\{1,2\}$.
In total, we evaluated $672$ test problems.

The repository to reproduce the benchmarks is referenced in the appendix.
All formulations have been implemented for the power system analysis toolbox PyPSA \cite{pypsa}
and the optimization problems are solved using the commercial solver Gurobi (version 9.0), given a time limit of $6$ hours each.
Primal simplex, dual simplex and interior point algorithms are run in parallel for each problem.
The solutions and solving times are retrieved from the fastest algorithm.

\begin{figure}
	(a) Europe: 1000 nodes \hspace{1.4cm} (b) Europe: 750 nodes
	\centering
	\includegraphics[width=0.48\columnwidth]{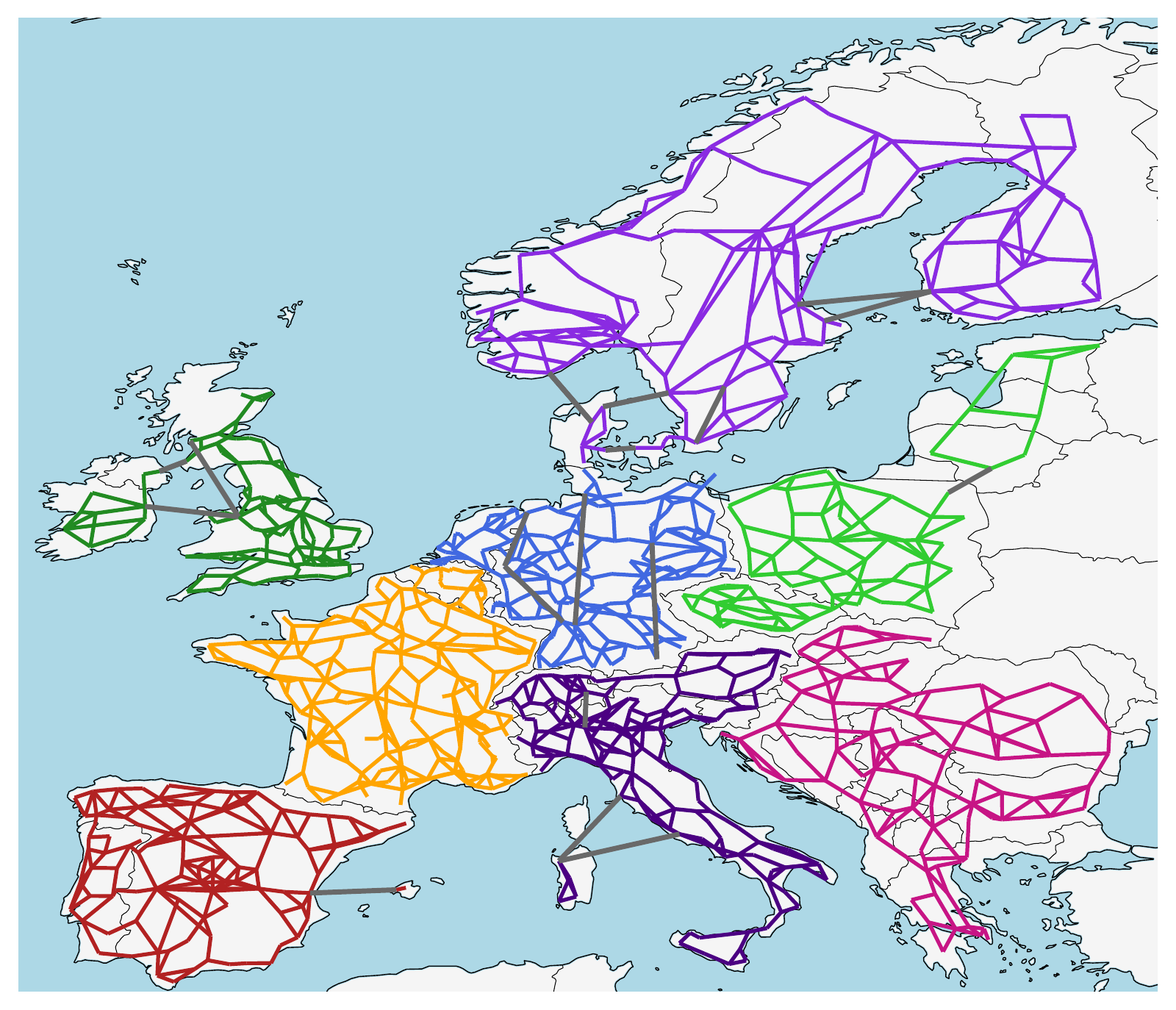}
	\includegraphics[width=0.48\columnwidth]{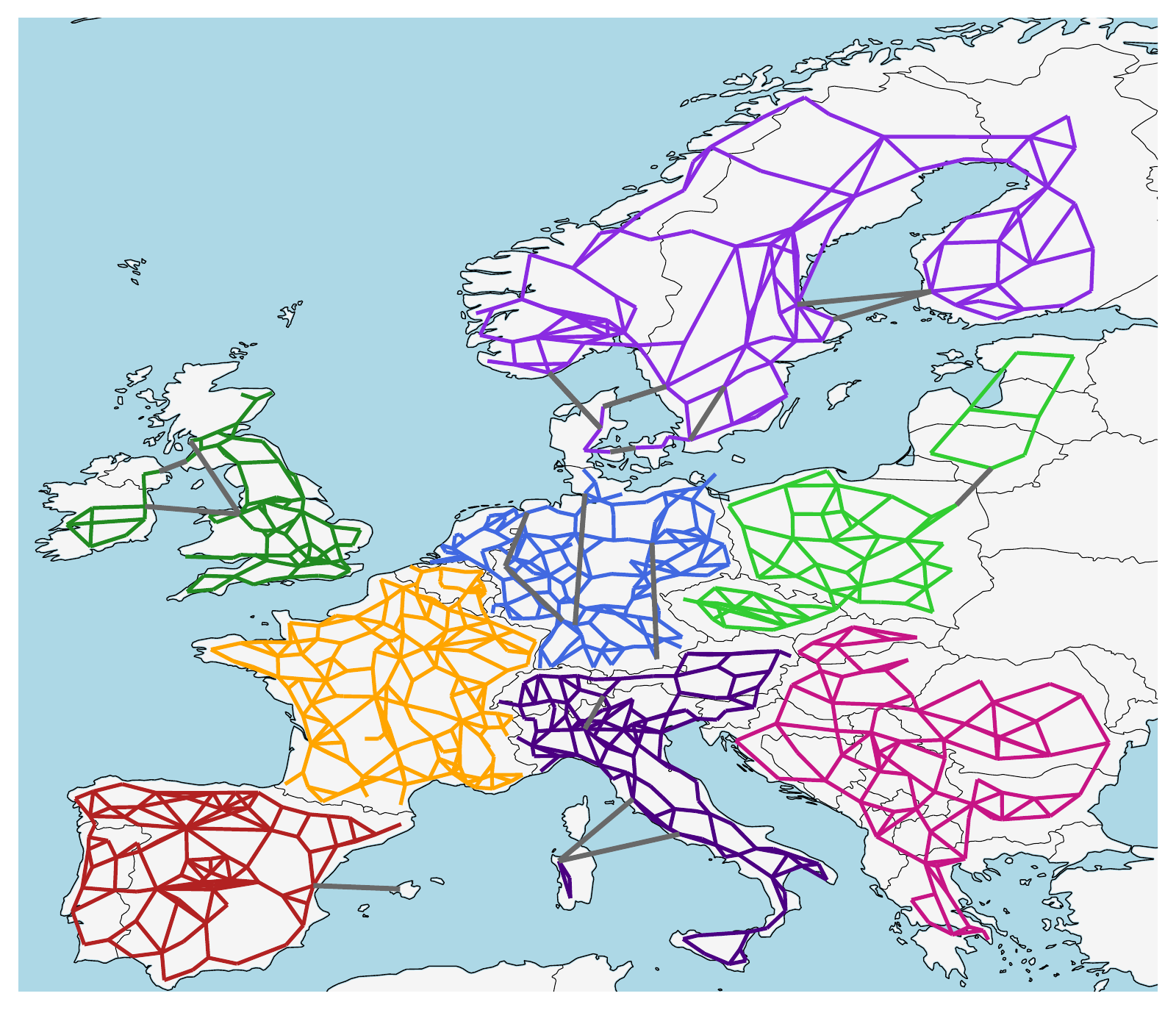}
	(c) Europe: 500 nodes \hspace{1.5cm} (d) Europe: 250 nodes
	\includegraphics[width=0.48\columnwidth]{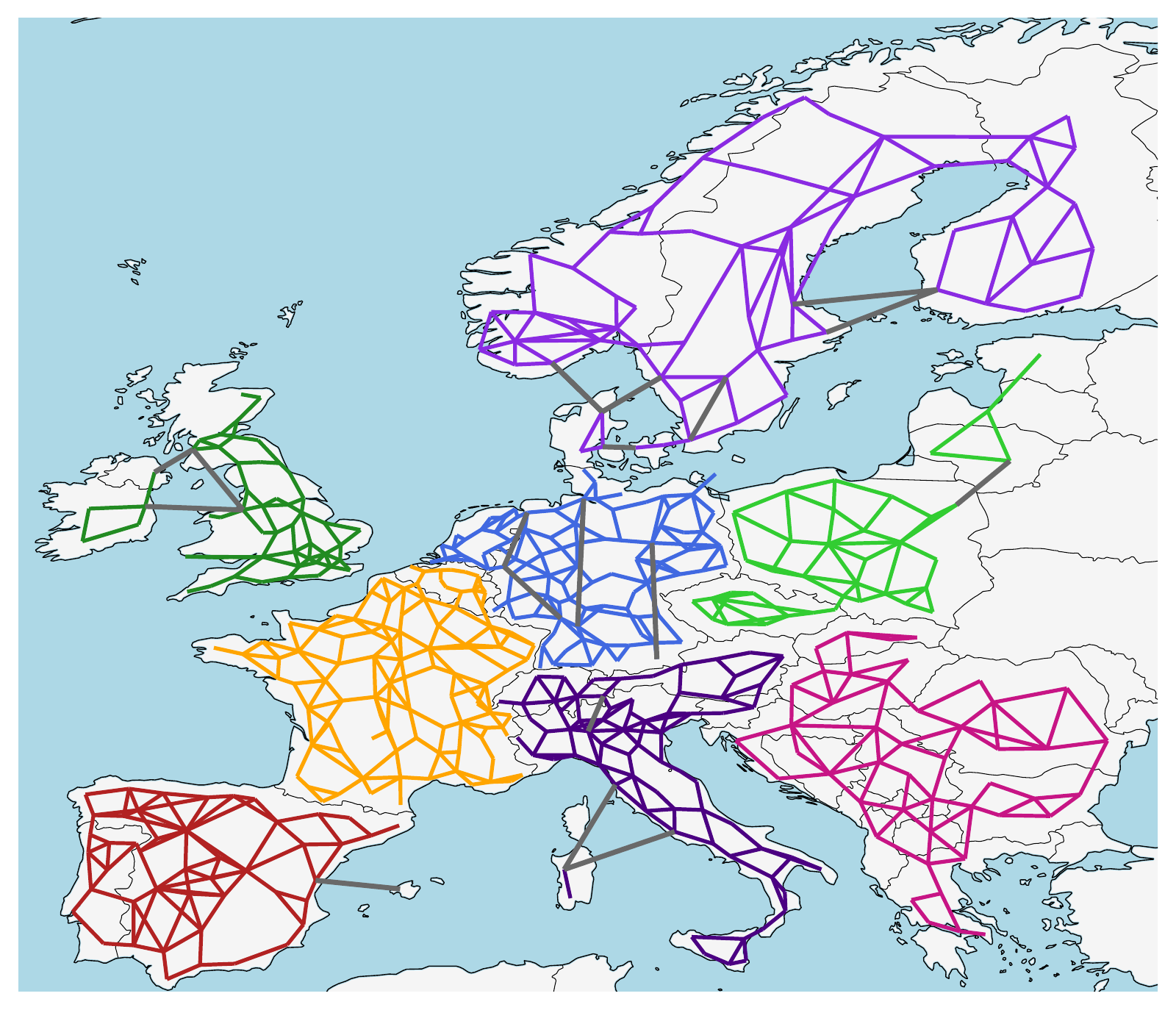}
	\includegraphics[width=0.48\columnwidth]{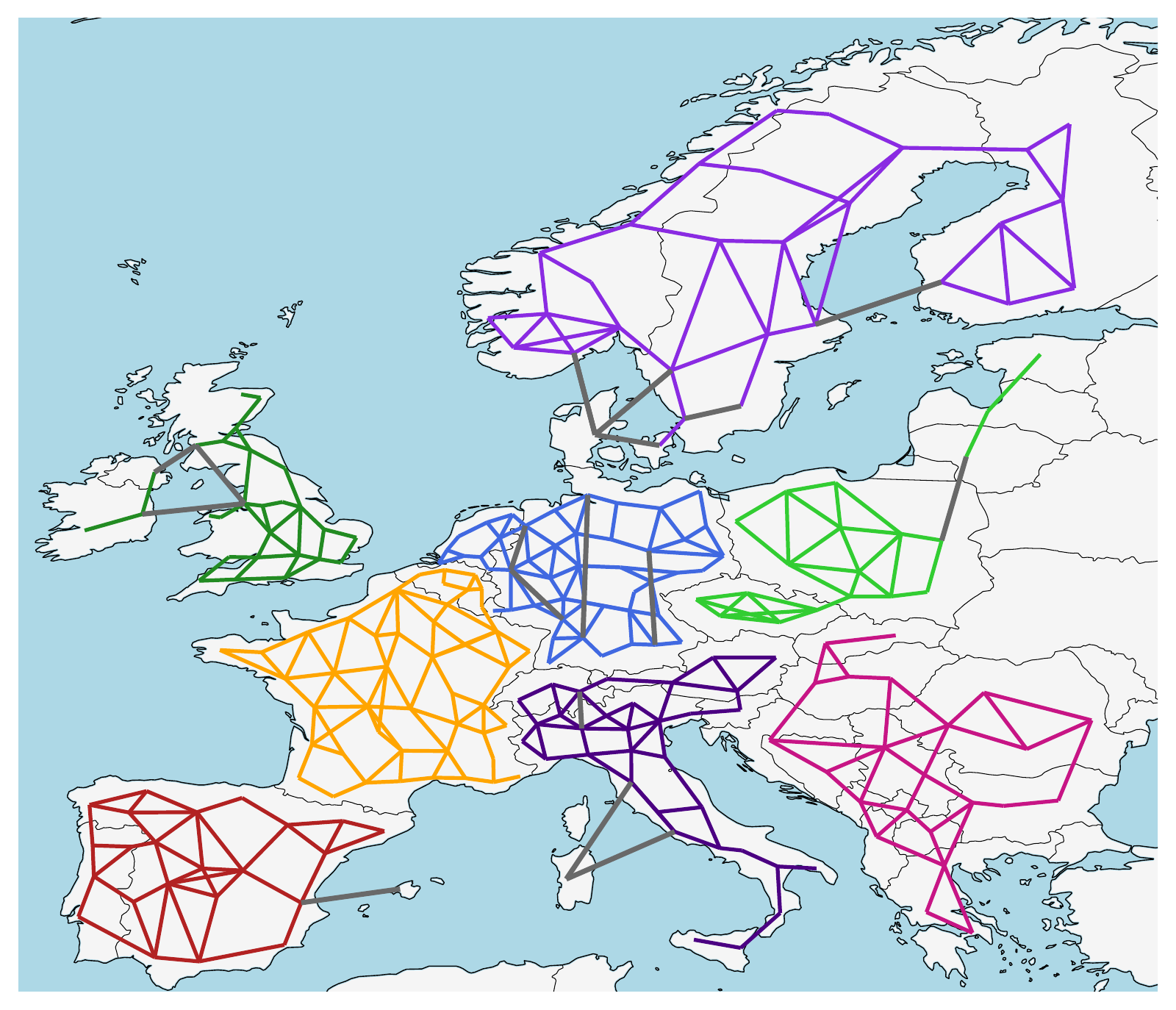}
	\caption{Clustered European transmission network models from which regional extracts are formed for the benchmark cases.
		Each color denotes a region.
		Colored lines represent AC transmission lines at 380 kV, gray lines represent HVDC links.}
	\label{fig:geography}
\end{figure}

\section{Results}
\label{sec:results}

To begin with, Figure \ref{fig:histograms} provides an initial insight on the problem sizes of the benchmark cases.
The benchmark set covers a wide range of many smaller and some more complex problems.
The largest involve up to 150,000 variables and 300,000 constraints and are the main target of speed improvements.
The number of binary investment variables ranges from 34 to 612 candidate lines.

\begin{figure}
	\includegraphics[width=\columnwidth]{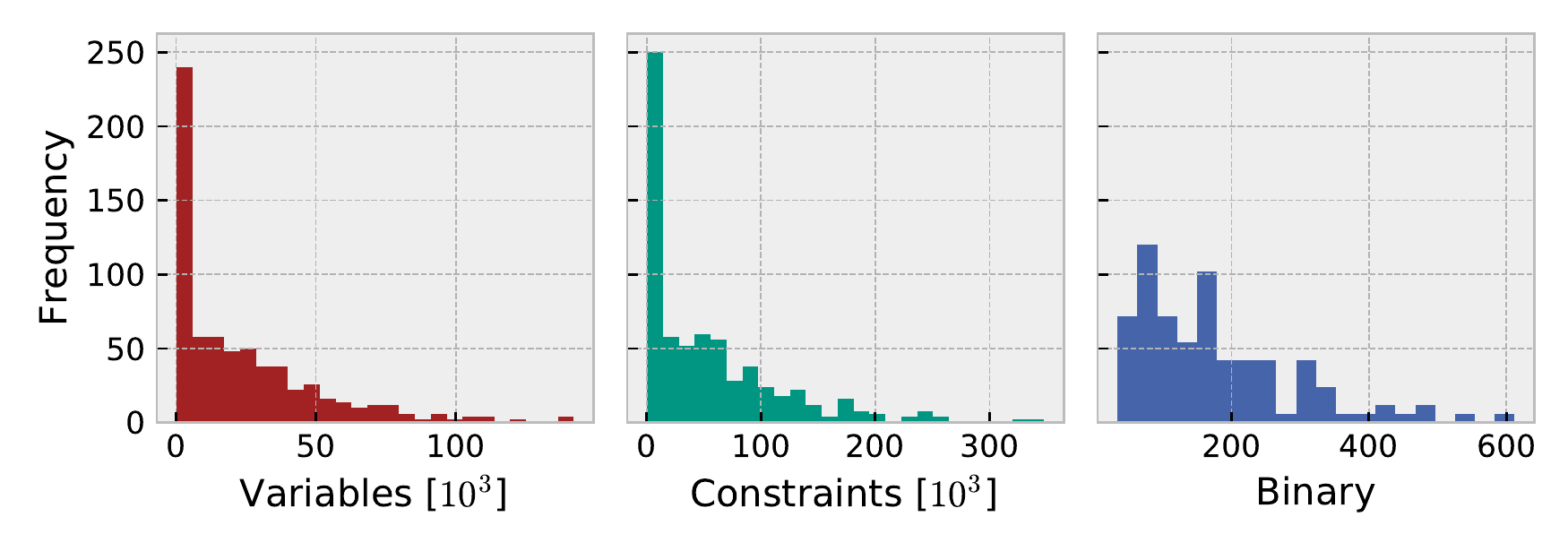}
	\caption{Histograms of the distribution of the number of variables, constraints and binary variables across the benchmark cases for the angle-based joint generation and transmission expansion planning.}
	\label{fig:histograms}
\end{figure}

On average, using the cycle-based formulation reduces the number of constraints to $95.3\%$
compared to the angle-based formulation. Likewise, the average number of variables is reduced to $90.5\%$.
As previously noted, this is due to the absence of the auxiliary voltage angle variables and fewer KVL constraints
in the cycle-based formulation.

A share of 92\% of all cycle-based problems and 82\% of all angle-based problems
were solved fastest using interior-point algorithms.
Otherwise, dual simplex was quickest.
To verify the formulations' objective values, while accounting for the fact that the
MILPs only solve up to a predefined tolerance,
we assert that the upper bound of one formulation is always larger than the lower bound of the other.
Across all instances the total volume of transmission expansion ranges between $0\%$ and $60\%$ of the existing
transmission network with up to $24$ TWkm of additional network capacity.
Due to the tolerances regarding the MIP gaps, both formulations can still yield slightly different
transmission expansion plans.

\begin{table}
	\begin{tabular}{@{}lrrr@{}}
		\toprule
		             & both         & at least one $>2$ min, & one unsolved \\
		             & $\leq 2$ min & excluding unsolved     & in walltime  \\
		\midrule
		instances    & 400          & 186                    & 26           \\
		share faster & 63.8\%       & 87.6\%                 & 100\%        \\
		\midrule
		\multicolumn{4}{l}{speed-up factor (angle-based / cycle-based):}    \\
		\midrule
		-- mean      & 1.28         & 3.94                   & 3.12         \\
		-- median    & 1.09         & 2.20                   & 2.11         \\
		-- maximum   & 9.40         & 31.25                  & 14.70        \\
		-- minimum   & 0.67         & 0.38                   & 1.06         \\
		\bottomrule
	\end{tabular}
	\caption{Numerical results for comparing the novel cycle-based to the standard angle-based formulation.
		The speed-up factor is calculated by dividing the solving time of the angle-based formulation
		by the solving times of the cycle-based formulation.}
	\label{tab:results}
\end{table}

In terms of computation times, the competing formulations are contrasted in Figure \ref{fig:times}
and Table \ref{tab:results}.  For individual benchmark cases
the relation between solving times is visible in Figure \ref{fig:times}.
If a point is located on the identity line, both angle-based and cycle-based formulation took the same period of time
to solve. If a point lies in the upper-left triangle the cycle-based formulation solved faster,
while a point in the lower-right triangle indicates that the angle-based formulation was quicker.
Some instances have slightly exceeded the walltime of $6$ hours due to system latency.

From Figure \ref{fig:times} it becomes clear that the cycle-based formulation has a distinct advantage
over the angle-based formulation in terms of computation times.
In $60$ out of the $672$ instances both formulations did not satisfy the required MIP gap within the time limit.
For the remaining instances we distinguish the cases
(i) both formulations solved in less than two minutes,
(ii) at least one formulation took more than two minutes but neither ran into the walltime, and
(iii) exactly one formulation did not solve within the time limit.
These categories are reflected in the summary of computational performance in Table \ref{tab:results}.

Instances of the most relevant group (ii), solve up to $31.25$ times faster for particular cases,
while averaging at a speed-up of factor $3.94$ when using the cycle-based formulation instead of the angle-based variant.
The median speed-up is $2.20$.
The angle-based formulation is outperformed in most (but not all) cases.
Only in $12.4\%$ of all cases, the angle based formulation was faster.
For 26 instances, one formulation could not satisfy optimality tolerances within the time limit of 6 hours.
In all such cases, the cycle-based formulation was solved, taking on average just 2 hours.
The reduced computational advantage for small problems can partially be explained by
the overhead that originates from determining the cycle basis and candidate cycles when building the problem.
\begin{figure}
	\centering
	\includegraphics[width=\columnwidth]{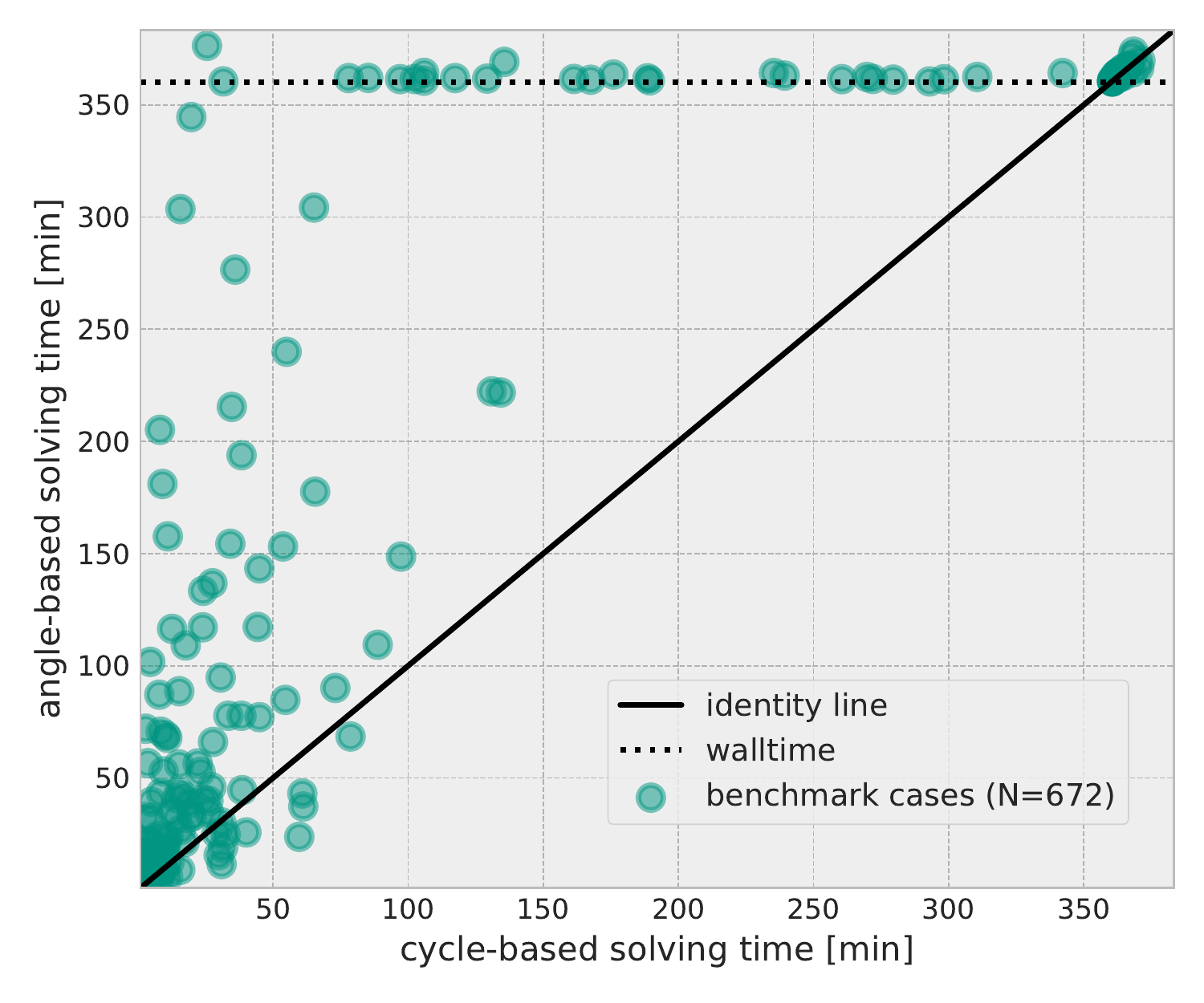}
	\caption{Solving times of cycle-based formulation versus solving times of angle-based formulation.}
	\label{fig:times}
\end{figure}

\begin{figure}
	\centering
	\includegraphics[width=\columnwidth]{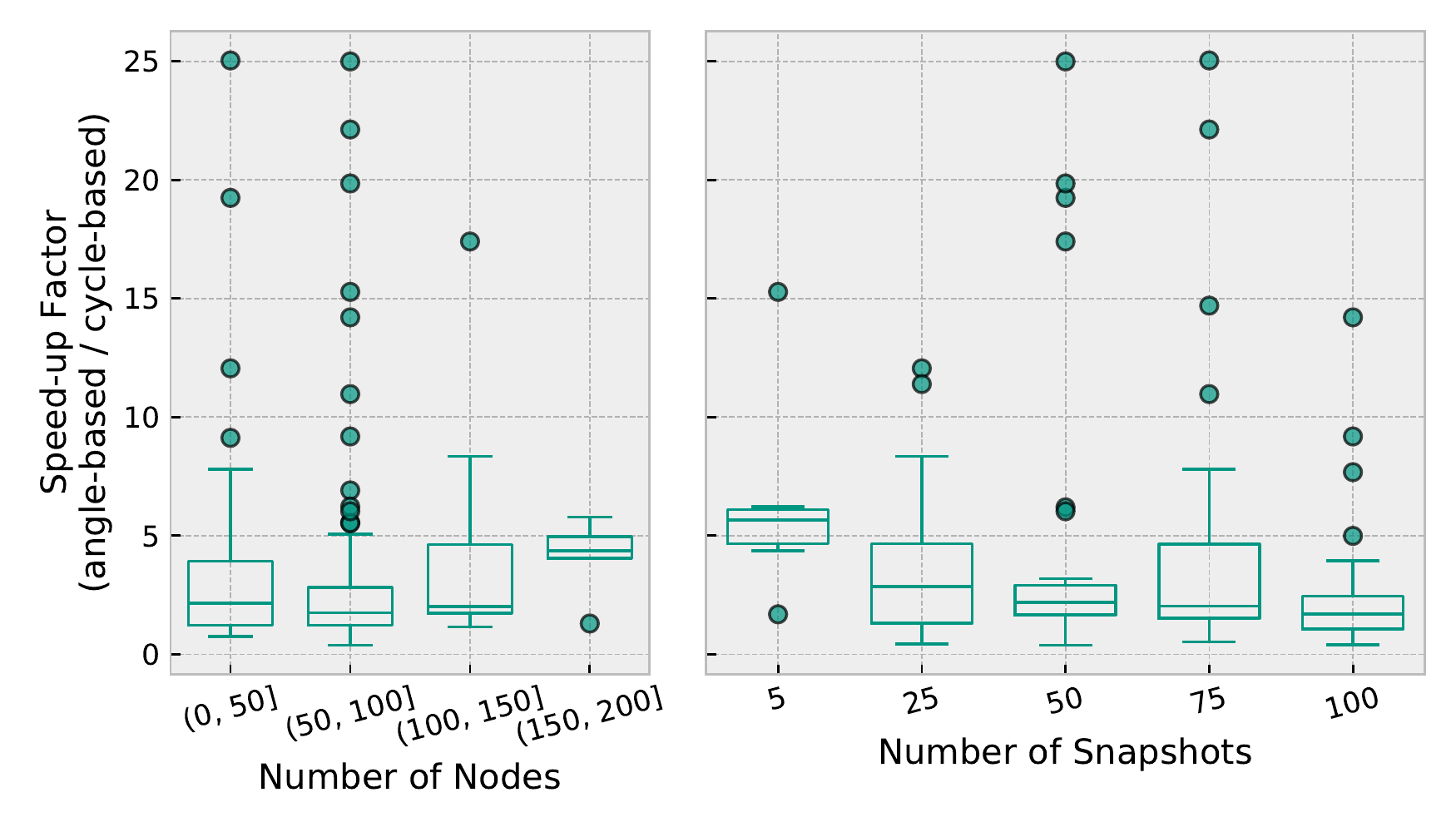}
	\caption{Sensitivities of speed-up factor distribution towards number of nodes and number of snapshots depicted as boxplots.}
	\label{fig:boxplots}
\end{figure}

\begin{figure}
	\centering
	(a) by required MIP Gap
	\includegraphics[width=\columnwidth]{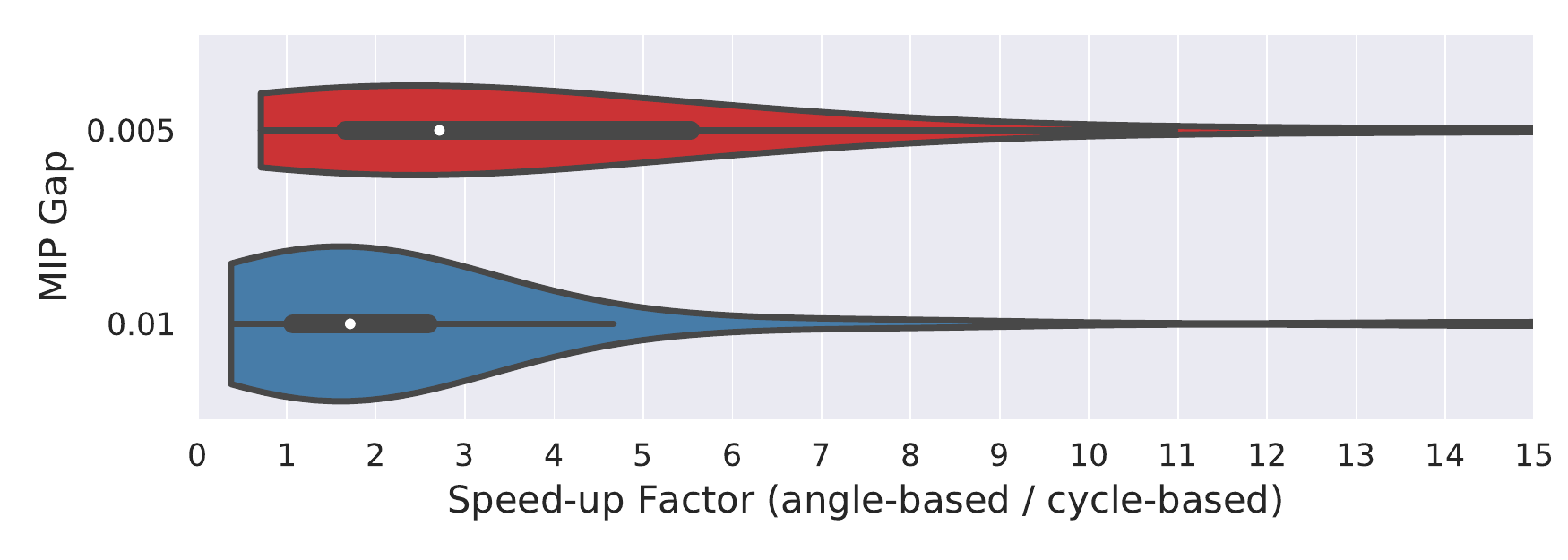}
	(b) by number of candidate lines per existing corridor
	\includegraphics[width=\columnwidth]{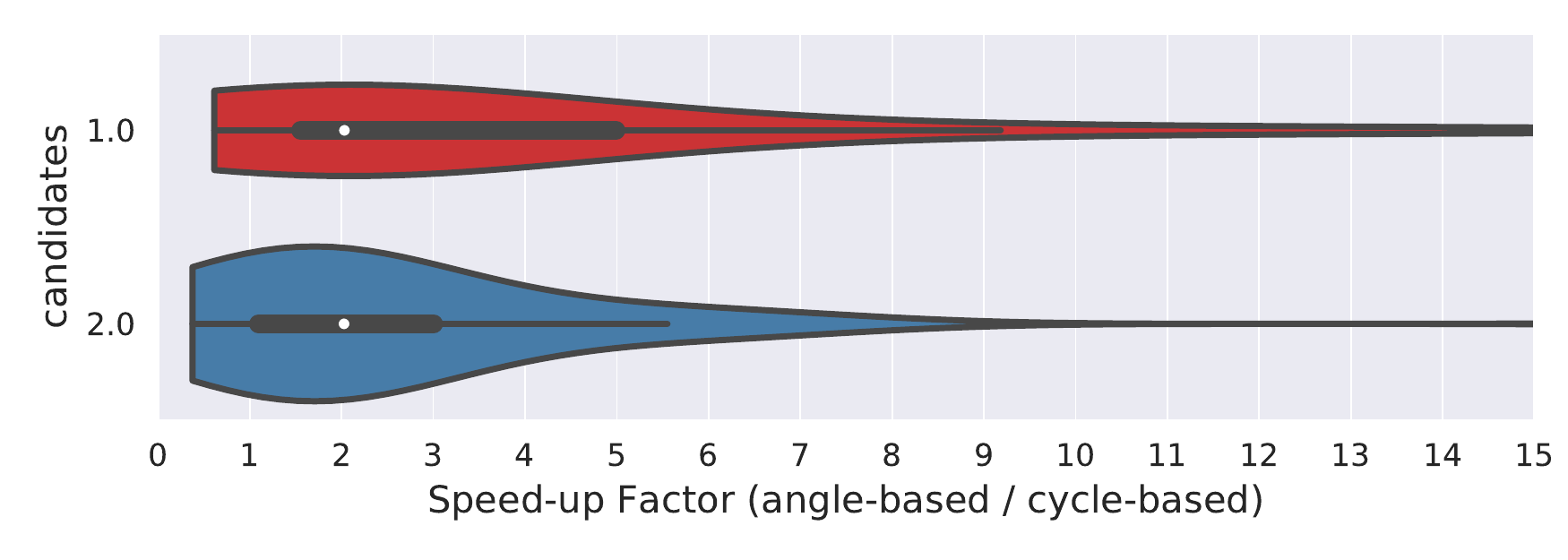}
	\caption{Sensitivities of speed-up factor distribution}
	\label{fig:violin}
\end{figure}

\begin{figure}
	\centering
	\includegraphics[width=\columnwidth]{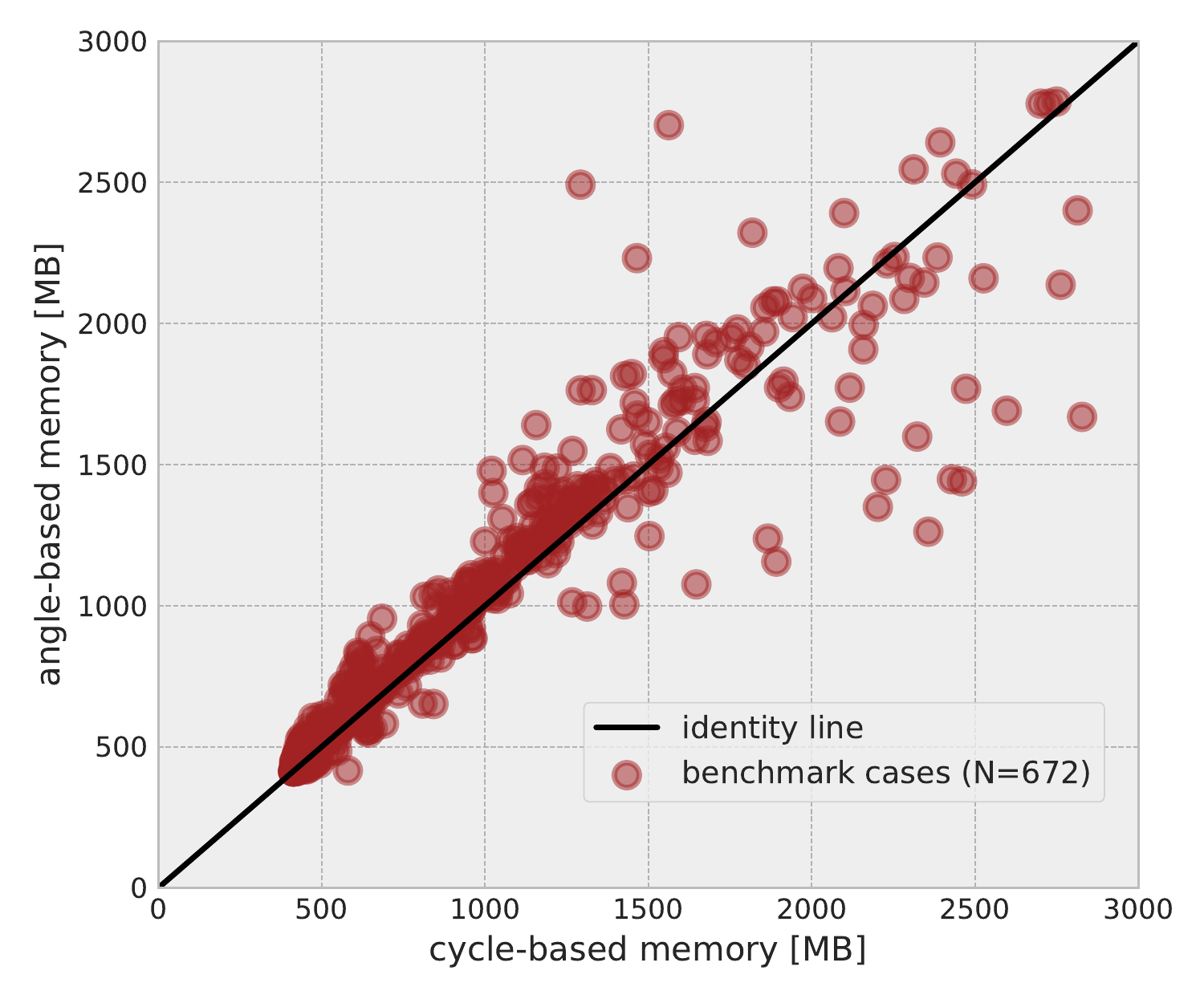}
	\caption{Memory consumption of cycle-based formulation versus solving times of angle-based formulation.}
	\label{fig:memory}
\end{figure}

The boxplots in Figure \ref{fig:boxplots} afford another interesting view on the solving times.
They show the sensitivities of speed-up factors towards the spatial and temporal resolution
of the network models.
Besides many outliers, a trend towards a higher acceleration with larger networks is visible.
Although acceleration tends to decrease to some extent with higher temporal resolution,
the cycle-based formulation still outperforms the angle-based formulation considerably in most cases.
Figure \ref{fig:violin} exhibits two further sensitivities.
We found that a tighter MIP optimality gap further develops the advantage of the cycle-based formulation,
while considering slightly more candidate lines for each existing corridor tends to reduce its benefit.
Contrary to computation times, as is shown in Figure \ref{fig:memory}, there is no clear preference for
either formulation in terms of peak memory consumption.

The fact that already for problems with few snapshots considerable speed-ups could be achieved,
makes the cycle-based reformulation also suitable for combining it with Benders decomposition or related decomposition schemes.
We did not apply any decomposition in this paper because the merits of Benders decomposition may be restricted to TEP problems
where there are no complicating time-dependent constraints, e.g.\ from storage consistency equations or carbon budgets.
Such intertemporal coupling would prohibit other essential acceleration techniques \cite{lumbreras_faster},
but is pivotal to factor in the multitude of tradeoffs in designing highly integrated renewable energy systems by co-optimization.

\section{Conclusion}
\label{sec:conclusion}

This paper developed a novel cycle-based reformulation for the transmission expansion planning (TEP)
problem with LOPF and compared it to the standard angle-based formulation.
Instead of introducing a large number of auxiliary voltage angle variables,
the cycle-based formulation expresses Kirchhoff's voltage law directly in terms of the power
flows, based on a cycle decomposition of the network graph. This results in fewer variables and sparser constraints.
The angle-based formulation, moreover, has the disadvantage that it is not well-suited to considering the connection of multiple
disconnected networks. The cycle-based formulation is shown to conveniently accommodate such synchronization options.
Since both formulations use the big-$M$ disjunctive relaxation, helpful derivations
for suitable big-$M$ values were provided to avert numerical problems.
The competing formulations were benchmarked on $672$ realistic generation and transmission expansion problems
built from an open model of the European transmission system.
For computationally challenging problems, the cycle-based formulation was shown to solve up to 31 times
faster for particular cases, while averaging at a speed-up of factor 4.
Hence, the cycle-based formulation is convincing not only because it can efficiently address synchronization options,
but also for its computational performance.

\begin{acks}
 F.N. and T.B. gratefully acknowledge funding from the Helmholtz Association
 under grant no. VH-NG-1352. F.N. also gratefully acknowledges funding from
 the Karlsruhe House of Young Scientists (KHYS) through the networking grant programme.
 The responsibility for the contents lies with the authors.
\end{acks}

\bibliographystyle{ACM-Reference-Format}
\bibliography{cycleflows}

\appendix

\section{Online Resources}

The code to reproduce the experiments of this paper is available at \url{https://github.com/fneum/benchmark-teplopf}.
The repository also contains the results as raw data.
The implementation of the transmission expansion planning problem in PyPSA can be found at \url{https://github.com/pypsa/pypsa/tree/tep-v2}.
Code and documentation of PyPSA-Eur are provided at \url{https://github.com/pypsa/pypsa-eur} and \url{https://pypsa-eur.readthedocs.io}.

\end{document}